\documentclass[aps,prx,noeprint,twocolumn,superscriptaddress,balancelastpage,10pt]{revtex4-2}
\usepackage{mysty}

\let\oldaddcontentsline\addcontentsline
\renewcommand{\addcontentsline}[3]{}

\begin{document}

\title{Quantum Mpemba Effect in Non-Equilibrium Quantum Thermometry}

\author{Zi-Shen Li}
\email{zishen@connect.hku.hk}

\author{Yuxiang Yang}
\email{yuxiang@cs.hku.hk}
\affiliation{Quantum Information and Computation Initiative, Department of Computer Science, School of Computing and Data Science, The University of Hong Kong, Pokfulam Road, Hong Kong, China}

\begin{abstract}
    The quantum Mpemba effect (QMpE) describes an anomalous thermalization phenomenon in which quantum states initially far from equilibrium can approach thermal equilibrium faster than states that begin closer to it. While this effect has been extensively studied in various frameworks, its practical implications for quantum information processing remain largely unexplored. We investigate the relationship between QMpE and quantum thermometry, focusing on non-equilibrium scenarios where measurements are performed during early-stage thermalization. In a Markovian model, we rigorously prove that the initial states that are optimal for thermometry exhibit QMpE with high probability and thermalize faster than most initial states. Our results reveal a fundamental connection between quantum thermodynamics and thermometry, suggesting that QMpE can be harnessed to enhance temperature estimation with quantum probes.
\end{abstract}

\maketitle

\section*{Introduction}
The Mpemba effect, a counterintuitive phenomenon in which hot liquid freezes faster than cold liquid, was first systematically documented by E. B. Mpemba and D. G. Osborne in the 1960s \cite{Mpemba1979cool}. Subsequently, a formalism for analyzing eigenmodes in Markovian stochastic processes was developed \cite{lu2017NonequilibriumThermodynamicsMarkovianMpembaEffectIts,klich2019MpembaIndexAnomalousRelaxation}, substantially advancing the theoretical understanding of this effect. Recent research has extended this framework to quantum systems \cite{Carollo2021ExponentiallyAcceleratedApproach,Kochsiek2022AcceleratingApproachDissipative,Moroder2024MpembaEffectThermo,zhang2025ObservationQuantumStrongMpembaEffect,ares2025QuantumMpembaEffects,murciano2024EntanglementAsymmetryQuantumMpembaEffectXY,joshi2024ObservingQuantumMpembaEffectQuantumSimulations,wang2024MpembaEffectsNonequilibriumOpenQuantumSystems,summer2026ResourceTheoreticalUnificationMpembaEffectsClassicalQuantum}, leading to the identification of the quantum Mpemba effect (QMpE). This quantum analog occurs when states initially distant from thermal equilibrium relax to equilibrium more rapidly than states initially closer to equilibrium, as illustrated in \cref{fig-intro}(a). The QMpE has been rigorously examined across diverse regimes, encompassing Markovian \cite{Moroder2024MpembaEffectThermo} and non-Markovian dynamics \cite{strachan2025NonMarkovianQuantumMpembaEffect}, random quantum circuits \cite{turkeshi2025QuantumMpembaEffectRandomCircuits,qian2025IntrinsicQuantumMpembaEffectMarkovian,liu2024SymmetryRestorationQuantumMpembaEffectSymmetric}, and many-body quantum systems \cite{joshi2024ObservingQuantumMpembaEffectQuantumSimulations,murciano2024EntanglementAsymmetryQuantumMpembaEffectXY}. 

Despite the rapidly growing literature on the QMpE, most existing works characterize it as a property of relaxation curves or spectral overlaps of a Liouvillian, while leaving open how such anomalous relaxation translates into operational advantages in quantum information tasks. 
A natural application of the Mpemba effect is thermometry \cite{DePasquale2018,kiilerich2018DynamicalApproachAncillaassistedQuantumThermometry,pasquale2017EstimatingTemperatureSequentialMeasurements,seah2019CollisionalQuantumThermometry,jarzyna2015QuantumInterferometricMeasurementsTemperature,fujiwara2021DiamondQuantumThermometryFoundationsApplications,campbell2018PrecisionThermometryQuantumSpeedLimit,Correa2015IndividualQuantumProbes,mehboudi2019ThermometryQuantumRegimeRecentTheoreticalProgress,cavina2018BridgingThermodynamicsMetrologyNonequilibriumQuantumThermometry,aiache2024NonMarkovianEnhancementNonequilibriumQuantumThermometry,jevtic2015SinglequbitThermometry,yu2024CriticalityenhancedPrecisionPhaseThermometry,boeyens2023ProbeThermometryContinuousMeasurements,stace2010QuantumLimitsThermometrya,srivastava2025TopologicalQuantumThermometry,rubio2021GlobalQuantumThermometry,mok2021OptimalProbesGlobalQuantumThermometry}. 
In classical thermodynamics, temperature is defined via the zeroth law: two systems in thermal equilibrium share the same temperature \cite{kittel1969ThermalPhysics}. 
The Mpemba effect offers a distinct advantage in classical thermometry by accelerating thermalization, thereby reducing the time required to attain the target temperature. In quantum thermometry, however, temperature can be estimated before the probe reaches equilibrium [\cref{fig-intro}(b)].
Such partial-thermalized thermometry is also known as non-equilibrium quantum thermometry \cite{Correa2015IndividualQuantumProbes,mehboudi2019ThermometryQuantumRegimeRecentTheoreticalProgress,aiache2024NonMarkovianEnhancementNonequilibriumQuantumThermometry,cavina2018BridgingThermodynamicsMetrologyNonequilibriumQuantumThermometry,jevtic2015SinglequbitThermometry}. 
This raises an ambiguity: faster relaxation could either help by amplifying temperature-dependent transitions early or hurt by washing out temperature information more quickly.
This motivates the central question of this work: what is the operational relation between the QMpE and the sensitivity of non-equilibrium quantum thermometry?

\begin{figure}
    \centering
    \includegraphics[width=0.95\columnwidth]{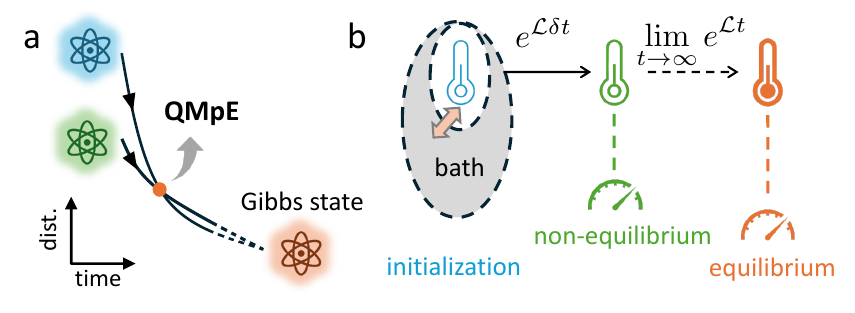}
    \caption{Schematic illustration of QMpE and thermometry. (a) Given a state initially further from thermal equilibrium in terms of a distance measure, QMpE occurs when it surpasses a state initially closer to equilibrium after finite time. (b) Illustration of non-equilibrium versus equilibrium quantum thermometry.}
    \label{fig-intro}
\end{figure}

Establishing a rigorous link requires identifying a regime where (i) thermometric optimality admits an analytic characterization and (ii) the resulting optimal probe preparation can be compared against typical states under a precise definition of QMpE. 
In this work, we consider a probe system with provably optimal thermometric performance and perform general measurements after short-time evolution, a strategy that is optimal when the total interrogation time is limited \cite{Correa2015IndividualQuantumProbes}.
From this, we obtain our main results: (i) we identify the optimal initial state that yields the maximum local distinguishability, and (ii) under this condition the QMpE necessarily appears with an exceptional probability exponentially small in the probe dimension $d$. 
The local distinguishability in this context is characterized by the trace distance between adjacent states within the parameterized family.
The results suggest that thermometric optimality implies an anomalously fast approach to equilibrium.

Our results suggest that the information-theoretic optimality of the input state at the short-time limit is associated with faster-to-equilibrium behavior in the thermalization procedure.
Numerical results further demonstrate that this correspondence persists in the finite-time regime.
Our work not only provides valuable insight into the QMpE, but also motivates further investigation of nonequilibrium thermodynamics from an information-theoretic perspective.
\medskip

\section*{Results}
\noindent\emphbold{Setup and optimal probe system. }
We can gain initial intuition about quantum thermometry and thermalization by recalling how temperature enters nonequilibrium dynamics.
Although temperature is fundamentally an equilibrium notion, in quantum statistical mechanics a thermal environment at inverse temperature $\beta$ is microscopically characterized by the Kubo-Martin-Schwinger (KMS) condition on its equilibrium correlation functions \cite{kubo1957StatisticalMechanicalTheoryIrreversible,Martin1959TheoryManyParticleSystems,haag1967EquilibriumStatesQuantumStatisticalMechanics,davies1974MarkovianMasterEquations,breuer2002TheoryOpenQuantum}.
This condition constrains the bath spectral correlations and thereby determines how temperature enters the reduced dynamics of the probe system.
In particular, after the weak-coupling Markovian, or Davies, reduction, the KMS structure is inherited by the probe through the temperature-dependent microscopic transition rates appearing in the dissipator \cite{spohn1978IrreversibleThermodynamicsQuantumSystems,AlickiLendi2007QuantumDynamicalSemigroups}.
Specifically, for a general $d$-level probe with system Hamiltonian
$\bH = \sum_{j=0}^{d-1} \hbar \omega_j \ketbra{j}{j}$
coupled to a heat bath at inverse temperature $\beta$, the reduced dynamics are governed by the Lindblad master equation
\cite{gorini1976CompletelyPositiveDynamicalSemigroups,lindblad1976GeneratorsQuantumDynamicalSemigroups}.
\begin{align}
    \frac{\d\rho}{\d t}
    = \cL_\beta[\rho]
    = - i[\bH,\rho] + \cD_\beta[\rho],
\end{align}
where $\cD_\beta$ contains the temperature-dependent dissipative transitions.\footnote{We do not use the Lamb shift as a thermometric resource, since its temperature-dependent phase contribution is sensitive to microscopic bath details and can be difficult to distinguish from detuning or other Hamiltonian frequency drifts. Unlike the dissipative rates, it is not fixed solely by KMS detailed balance \cite{breuer2002TheoryOpenQuantum,AlickiLendi2007QuantumDynamicalSemigroups}.}
To ensure thermodynamic consistency and identify the probe geometry relevant to optimal thermometry, we impose two physical assumptions:

{\emph{(i) Weak-coupling thermalization (Davies map).}} {We assume the probe interacts with the heat bath} in the weak-coupling, time-homogeneous regime. {Under this assumption, the dissipative dynamics are characterized by the Davies map \cite{davies1979GeneratorsDynamicalSemigroups,lostaglio2018ElementaryThermalOperations}, which strictly enforces thermodynamic consistency.} The generator of the Davies map exhibits two properties. First, the superoperators corresponding to its unitary and dissipative components commute. Consequently, the dissipative generator takes the form
\begin{align}
    \cD_\beta[\rho] &= \sum_{\substack{i,j\in[d]\\i>j}} \Big[ \Gamma_\beta(\omega_j-\omega_i)\cD^{(i\rightarrow j)}[\rho] \nonumber \\
    &\quad + \Gamma_\beta(\omega_i-\omega_j)\cD^{(j\rightarrow i)}[\rho] \Big],
\end{align}
where $\cD^{(i\rightarrow j)}$ represents the dissipator governing the transition between the $i$-th and $j$-th energy levels, and $\Gamma_{\beta}(\omega)$ is the transition weight. Second, these transition weights satisfy the detailed balance condition \cite{alicki1976DetailedBalanceCondition}, meaning
\begin{align}\label{eq:kms-detail}
    \frac{\Gamma_\beta (\omega)}{\Gamma_\beta (-\omega)} = e^{-\beta \omega}.
\end{align}
{The specific form of $\Gamma_\beta(\omega)$ depends on the system--bath interaction.} It is convenient to rewrite the transition rates as $\Gamma_\beta(\omega)=J(\omega)f_\beta(\omega)$, where $f_\beta(\omega):=\left[\exp(\beta\omega)-1\right]^{-1}$ denotes the mean thermal occupation number of bosonic bath excitations \cite{DePasquale2018}.

{\emph{(ii) Effective two-band probe topology.}} While the system may generally exhibit an arbitrary energy-level coupling topology, i.e., transitions can happen between all $d$ levels, it is sufficient to consider an effective model consisting of a single ground state coupled to a $(d-1)$-dimensional excited manifold.
For thermometric purposes, not all the excitations $\{\cD_\beta^{(i\rightarrow j)}\}_{i\neq j}$ contribute high thermal sensitivity due to the following reasons.
In realistic settings, the transition weight $\Gamma_\beta$ is not a uniform function of frequency.
There typically exists an optimal frequency $\omega_{\rm opt}$ near which the probe system exhibits maximal thermal sensitivity \footnote{More specifically, one can tailor the energy level spacings to cluster around $\omega_{\rm opt} := \arg \max_{\omega} |\partial_\beta \Gamma_\beta(\omega)|$.}. Here, we set $\omega_0=0$ and assume the remaining frequencies $\{\omega_j\}_{j=1}^{d-1}$ are closely spaced with a small detuning $\varepsilon$, such that $|\omega_i-\omega_j| \le \varepsilon$ for all $i,j\in\{1, \dots, d-1\}$.
It has been shown that this configuration is optimal for equilibrium quantum thermometry \cite{Correa2015IndividualQuantumProbes}, where inserting any energy levels between them lowers the effective heat capacity, yielding a smaller thermal sensitivity.
From experimental perspective, such probe systems can be constructed in various platforms, including atomic systems \cite{anton2002OpticalBistabilityUsingQuantumInterferenceVtype,abdelaziz2020EffectiveControlSwitchingOpticalMultistabilityThreelevel,saffman2010QuantumInformationRydbergAtoms}, quantum dots \cite{kouwenhoven1997ExcitationSpectraCircularFewElectronQuantumDots}, superconducting circuits \cite{koch2007ChargeinsensitiveQubitDesignDerivedCooperPair}, and ion traps \cite{haffner2008QuantumComputingTrappedIons}.

\medskip

\noindent\emphbold{Characterizing thermal sensitivity. }
In the context of quantum metrology, the precision of parameter estimation is characterized by various bounds.
Recent studies have developed connections between parameter estimation and channel discrimination \cite{Tsang2012ziv-zakai,walter2014LowerBoundsQuantumParameterEstimation,berry2015QuantumBellZivZakaiBoundsHeisenbergLimitsWaveform,meyer2025QuantumMetrologyFiniteSampleRegime}.
In this work, we consider the local distinguishability as the figure of merit, which has been considered in various metrological setups such as binary hypothesis testing \cite{Tsang2012ziv-zakai,walter2014LowerBoundsQuantumParameterEstimation,berry2015QuantumBellZivZakaiBoundsHeisenbergLimitsWaveform}.
More specifically, given a continuous parameterized family of quantum channels, i.e., $\{\cE_\beta\}$, the local distinguishability given input state $\rho_0$ is defined as 
\begin{align}
    \lone{\partial_\beta \cE_\beta[\rho_0]}.
\end{align}
The connection between the local distinguishability and the estimation error can be demonstrated in various ways, such as Ziv-Zakai error bound \cite{Tsang2012ziv-zakai,berry2015QuantumBellZivZakaiBoundsHeisenbergLimitsWaveform} (see also the channel-discrimination perspective in Refs.~\cite{walter2014LowerBoundsQuantumParameterEstimation,meyer2025QuantumMetrologyFiniteSampleRegime}).

Regarding the thermometry problem, the parameterized channels are of the form $\{e^{\cL_\beta t}\}$.
Here the interrogation time $t$ is an important resource in non-equilibrium quantum thermometry \cite{Correa2015IndividualQuantumProbes}. 
Practical constraints often limit the allowable interrogation time. 
One may either estimate the temperature after total evolution time $t_s$ or partition it into $n$ independent runs, each of duration $\Delta t = t_s/n$. 
It has been demonstrated that optimal thermometric performance can only be achieved in the limit $\Delta t \rightarrow 0$ \cite{Correa2015IndividualQuantumProbes,chin2012QuantumMetrologyNonMarkovianEnvironments}. 
Therefore, we focus on the short-time regime where $\Delta t$ is considered as infinitesimal and will later provide numerical results in the finite-$\Delta t$ regime.
The temperature estimation problem in this case reduces to maximizing the following trace norm:
\begin{align}
    \max_{\varrho_0} \lone{\partial_\beta \cL[\varrho_0]}\label{eq:max-problem}.
\end{align}
The optimal initial state is denoted by $\rho^\star$, i.e., 
\begin{align}\label{eq:opt_state}
    \rho^\star:= \arg \max_{\varrho_0} \lone{\partial_\beta \cL [\varrho_0]}.
\end{align}
As the problem is not convex, finding the global optimum is generally challenging.
We will later introduce an analytical tight upper bound for the objective function that allows us to identify the optimal initial state.

\medskip

\noindent\emphbold{Optimal thermometry implies QMpE. }
We now compare the thermalization processes between the optimal initial state for thermometry $\rho^\star$, defined by Eq.~(\ref{eq:opt_state}), and a random initial state. 
Specifically, we say that $\varrho^{(1)}$ exceeds $\varrho^{(2)}$ in thermalization if there exists $t'\in[0,\infty)$ such that for all $t\ge t'$, the Frobenius distance between the first state and the thermal state remains consistently smaller than that of the second state, i.e., $\forall t\ge t',\fnorm{\varrho_{t}^{(1)} - \tau_\beta} \le \fnorm{\varrho_{t}^{(2)} - \tau_\beta}$, where both states evolve under the same dynamics, i.e., $\varrho_t^{(1)}= \exp(\cL t)[\varrho^{(1)}]$ and $\varrho_t^{(2)}= \exp(\cL t)[\varrho^{(2)}]$.
We note that the criterion for exceeding can be defined using alternative distance measures or even thermomajorization \cite{ruch1978MixingDistance,horodecki2013FundamentalLimitationsQuantumNanoscaleThermodynamics,marshall2011InequalitiesTheoryMajorizationApplications,streltsov2017ColloquiumQuantumCoherenceResource,chitambar2019QuantumResourceTheories,lipkabartosik2024CatalysisQuantumInformationTheory,brandao2015SecondLawsQuantumThermodynamics,lostaglio2015DescriptionQuantumCoherenceThermodynamicProcesses,vu2025ThermomajorizationMpembaEffect}. Here, we adopt the Frobenius distance for conciseness.
We present the following theorem establishing a crucial connection between optimal thermometry and QMpE:
\smallskip
\begin{theorem}\label{thm:prob-of-exceeding}
    Let $\rho_t^{\mathrm{ref}} = \exp(\cL t)[\rho^{\mathrm{ref}}]$ be the evolved state starting from a random reference state $\rho^{\mathrm{ref}} = (1-\alpha)\tau_\beta + \alpha \sigma$. 
    Here $\alpha \in (0,1]$ and $\sigma$ is a Haar random pure state ($\sigma = U \ketbra{0}{0} U^\dg,U\sim\mu_H$). 
    For $d\ge 3$, the probability of $\rho^{\mathrm{ref}}_t=\exp(\cL t)[\rho^{\mathrm{ref}}]$ being exceeded by $\rho^\star_t = \exp(\cL t)[\rho^\star]$ is at least $1-\delta$, with an exponentially decaying $\delta$:
    \begin{align}\label{eq:delta-exp-bound-main}
        \delta \le 2\exp\!\left(
        -\frac{d}{36\pi^3}\left[\frac{(d-2)(d-1)}{d(d+1)} - \frac{11\,\varepsilon^2}{10\alpha^2} g^2\right]^{\,2}
        \right).
    \end{align}
    Here $\rho^\star$ is the optimal state maximizing (\ref{eq:max-problem}), $g :=\left|\partial_\omega \log\left[e^{\omega\beta}\Gamma_\beta(\omega)\right]\right|_{\omega = \omega_{\rm opt}}$, and $M>0$ is a dimension-independent constant. For $d=2$, $\delta = 0$.
\end{theorem}
\smallskip

\Cref{thm:prob-of-exceeding} indicates that the optimal initial state for thermometry exhibits QMpE with probability $1-\exp(-\Omega(d))$ when compared to a random reference state, which is reminiscent of the quantum strong Mpemba effect \cite{zhang2025ObservationQuantumStrongMpembaEffect,furtado2025EnhancedQuantumMpembaEffectSqueezedThermal} where a state thermalizes faster than all other initial states.
It also reveals a non-typicality of the initial state in quantum information processing tasks, standing in contrast to the typicality phenomenon where relaxation dynamics become nearly initial-state-independent as system size increases \cite{bao2026InitialStateTypicalityQuantumRelaxation}.
A schematic illustration is provided in \cref{fig:traj-illustration}. The shaded ribbons representing the $\beta$-family trajectories of $\{\mathcal{L}_{\beta'}\}_{\beta' \in [\beta-\delta\beta,\beta+\delta\beta]}$ {gradually expand} throughout the entire evolution. 
Although they eventually converge to the same family of thermal states, given a finite interrogation time, the trajectories associated with $\rho^*$ expand more than those of $\rho^{\mathrm{ref}}$, yielding {enhanced} thermometric sensitivity. 
Moreover, despite being initially farther from $\tau_\beta$ than the reference state, {$\rho^\star$ reaches} the vicinity of thermal equilibrium more rapidly, thereby illustrating the quantum Mpemba effect in an operational thermometric setting.

\begin{figure}
    \centering
    \includegraphics[width=0.85\linewidth]{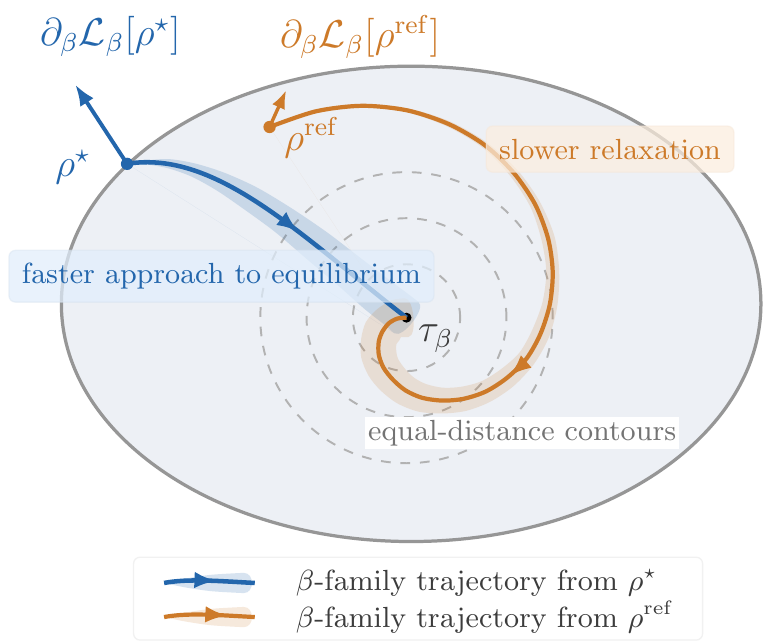}
    \caption{Conceptual illustration of the main result: thermometric optimality implies {accelerated convergence toward thermal equilibrium}. The blue trajectory starts from the optimal probe state $\rho^\star$, which maximizes the short-time thermometric sensitivity, $\|\partial_\beta \mathcal{L}_\beta(\rho)\|_1$. The orange trajectory starts from a reference state $\rho^{\mathrm{ref}}$, {whose} {sensitivity vector of smaller magnitude} indicates a weaker thermometric response. The shaded ribbons {characterize} the associated $\beta$-family trajectories, {representing} the smooth one-parameter families of states generated from $\rho^\star$ and $\rho^{\mathrm{ref}}$ under the continuous family of generators $\{\mathcal{L}_{\beta'}\}_{\beta' \in [\beta-\delta\beta,\beta+\delta\beta]}$.}
    \label{fig:traj-illustration}
\end{figure}

There are two important remarks regarding the theorem.
First, the conclusion holds for any fixed $\alpha$, including $\alpha = 1$ (i.e., when the initial state is a Haar random pure state).
The bound is non-trivial whenever the bracketed term in the exponent is positive.
In particular, this holds for $\alpha=\Omega(\varepsilon)$ when $g=O(1)$.
Second, the time at which the exceeding occurs is not specified; it could, in principle, be any finite time within the interval $[0,\infty)$.
In particular, if the exceeding occurs at $t=0$, it signifies that the optimal state for thermometry is already closer to the thermal state than a typical random state at the initial time—a scenario we include as a trivial case.
This trivial case can be excluded by choosing a smaller $\alpha$  so that $\rho^{\mathrm{ref}}$ is initially closer to $\tau_\beta$ than $\rho^\star$.

It is worth noting that the performance for thermometry is characterized by its short-time behavior as $\Delta t \rightarrow 0$, whereas QMpE is defined in the finite-time regime with $t\in[0,\infty)$.
Crucially, our results show that optimal short-time thermometric performance implies the occurrence of QMpE at finite time with probability at least $1-\exp(-\Omega(d))$.
However, the converse is not necessarily true.
In our setting, maximizing the local distinguishability selects a pure initial state, whereas the thermalization speed need not be monotone in purity.

\medskip
\noindent\emphbold{Sketch of the proof. }
First, we show how we can determine the optimal initial state for thermometry. 
We utilize the property of the Davies map that allows us to decompose the states into block-diagonal form. 
For convenience, we denote the population subspace by $\mathcal{P} := \text{span}\{\ket{j}\bra{j} \mid j=0,1,\ldots,d-1\}$ and the coherence subspace by $\mathcal{C} := \text{span}\{\ket{i}\bra{j} \mid i,j=0,1,\ldots,d-1, i\neq j\}$.
Since $\cL$ is of Davies form, it does not couple the population subspace with the coherence subspace, i.e., $\cL$ only maps from $\mathcal{P}$ to $\mathcal{P}$ and from $\mathcal{C}$ to $\mathcal{C}$.
This property permits the following decomposition.
Consider a general pure state $\varrho_0 = \ket{\psi}\bra{\psi}$ with $\ket{\psi} = \sqrt{\eta} \ket{0} + \sqrt{1-\eta} \sum_{j=1}^{d-1} c_j \ket{j}$, where $\eta\in[0,1]$ and $\sum_{j=1}^{d-1} |c_j|^2 = 1$. 
The operator $\partial_\beta \cL[\varrho_0]$ then decomposes as follows:
\begin{align}
    \begin{pmatrix}
        (1-\eta) A + \eta B& \sqrt{\eta(1-\eta)} \partial_\beta \cL[\ket{0}\bra{\til\psi}]\\
        \sqrt{\eta(1-\eta)} \partial_\beta \cL[\ket{\til\psi}\bra{0}] & (1-\eta)a + \eta b,
    \end{pmatrix}\label{eq:block-matrix}
\end{align}
where $A\oplus a$ constitutes the block-diagonal matrix defined by $\partial_\beta\cL[\ket{\til\psi}\bra{\til\psi}]$, $B\oplus b$ is the diagonal matrix equal to $\partial_\beta\cL[\ket{0}\bra{0}]$, and $\ket{\til\psi} = \sum_{j=1}^{d-1} c_j \ket{j}$.
To show that the ground state is optimal, it is equivalent to showing that the maximum of the trace norm of (\ref{eq:block-matrix}) is achieved at $\eta = 1$.
While in general $\eta = 1$ is not always the maximum point, there are certain conditions as demonstrated in \cref{lem:trace-norm-bound-2-main-text} (see Methods).
We show these conditions are satisfied by (\ref{eq:block-matrix}), and then its trace norm reduces to the following form:
\begin{align}\label{eq:main-text-optimal-probe-inequality-roof}
    \lone{\partial_\beta \cL [\varrho_0]} \le 2\left|\sum_{j=1}^{d-1}\frac{\partial}{\partial \beta}\Gamma_\beta(\omega_j)\right|.
\end{align}
The RHS of the bound is independent of $\varrho_0$, with equality being achieved if and only if $\varrho_0$ is the ground state.
By looking into the summation terms in \cref{eq:main-text-optimal-probe-inequality-roof}, we find that it is exactly the sum over all the weights of transitions that connect to the ground state.
This can be generalized to other energy-level coupling topologies, where, 
with eigenstate input, the short-time thermal sensitivity is always set by \cref{eq:main-text-optimal-probe-inequality-roof}.
This suggests that the optimal eigenstate is the one that has the most connectivity to other eigenstates.
In that case, the optimal state is not necessarily the ground state.

Second, we analyze how fast the ground state thermalizes compared to other states.
We do so by analytically solving the Liouvillian $\cL_\beta$.
The system state after evolution time $t$ can be expressed as follows \cite{Carollo2021ExponentiallyAcceleratedApproach,Moroder2024MpembaEffectThermo}
\begin{align}
    \varrho_t = e^{\cL_\beta t}[\varrho_0] = \hat r_1 + \sum_{i=2}^{d^2} e^{t \lambda_i} \Tr(\hat l_i^\dg \varrho_0) \hat r_i,
\end{align}
where $r_i$ and $l_i$ are the $i$th eigen-operators of $\cL_\beta$ and its dual map is defined as $\cL^+_\beta[\cdot]= i [\bH,\cdot]+\sum_j\left(L_j^\dg \cdot L_j-\{L_j^\dg L_j,\cdot\}/2\right)$.
$\lambda_i$ denotes the $i$th eigenvalue.
The eigenvalues and eigenoperators can be found by solving the eigenequations: $\cL[\hat r_i] = \lambda_i \hat r_i$ and $\cL^+[\hat l_i] = \lambda_i^* \hat l_i$.
Note that $\{r_i\}$ and $\{l_i\}$ themselves are not orthogonal since $\cL_\beta$ is not normal.

Importantly, the Davies map is naturally diagonalized within the coherence subspace $\mathcal{C}$.
We restrict our attention to the $d$-dimensional population subspace $\mathcal{P}$.
Directly solving $\cL_\beta$ is challenging, so we can first solve the zero-detuning case ($\varepsilon = 0$) and generalize the result using perturbation theory.
In the zero-detuning case, the map $\cL_\beta$ preserves the permutation symmetry among the excited states; consequently, the state $\rho_s := \sum_{j=1}^{d-1} \ket{j}\bra{j}/({d-1})$ and the ground state form a closed two-dimensional subspace under the dynamics.
This reduced system admits an exact analytical solution. 
It can be verified that one of the eigenoperators of $\cL_\beta$ within this subspace is the Gibbs state.
Furthermore, one can verify that the complementary subspace within $\mathcal{P}$ orthogonal to both $\ket{0}$ and $\rho_s$ constitutes a degenerate eigensubspace of $\cL_\beta$. 
Therefore, it is always possible to choose eigenoperators within this subspace that are mutually orthogonal.
As a result, the only pair of non-orthogonal eigenoperators arises from the two-dimensional subspace spanned by $\ket{0}$ and $\rho_s$.
This fact allows us to use the eigen-decomposition to precisely characterize the convergence behavior {via} \Cref{eq:convergence-bound-orthogonal-case}.
The squared Frobenius norm of the deviation from {the} steady state simplifies to:
\begin{align}\label{eq:convergence-bound-orthogonal-case}
    \fnorm{\varrho_t - \hat r_1}^2 &= \sum_{i=2}^{d^2} \left| \Tr(\hat l_i^\dg \varrho_0) \right|^2 e^{2\Re(\lambda_i) t}.
\end{align}
We then demonstrate that a small detuning constitutes only a minor perturbation, altering eigenoperators and eigenvalues by at most $O(\varepsilon)$.
Utilizing the results of the spectral analysis, we observe that the optimal state thermalizes via the fastest decay mode. Specifically, the distance to equilibrium scales as $\fnorm{\rho_t^\star - \hat r_1}^2 \sim e^{-2\Lambda_{\max}t}$, where $\Lambda_{\max} = \max_j \{|\Re(\lambda_j)|\}$.
Finally, we invoke concentration of measure for Haar-random pure states and obtain an exponentially small upper bound on the failure probability $\delta$. The full proof can be found in the Methods section.

\begin{figure}
    \centering
    \includegraphics[width=0.8\columnwidth]{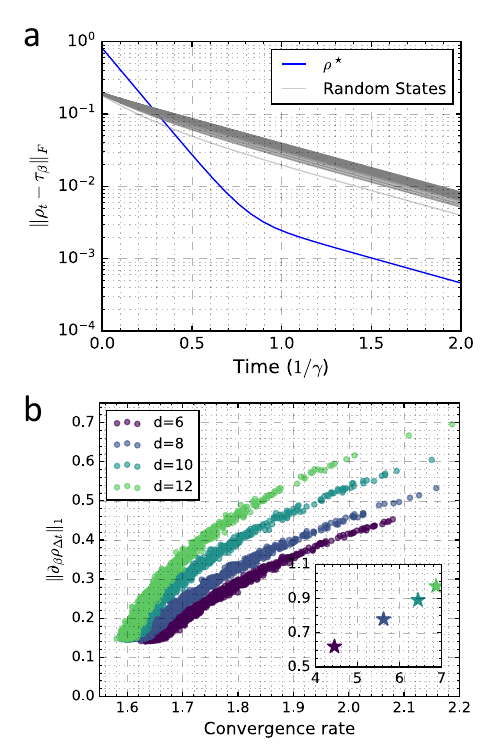}
    \caption{Numerical results. (a) The trace distance between the evolved state and the thermal state for two different categories of initial states: the optimal state for thermometry (ground state) and a random initial state. (b) The trace norm $\lone{\partial_\beta \brho_{\Delta t}}$ and the convergence rate. The convergence rate is the coefficient fitted through linear regression of $\log \fnorm{\brho_{t} - \tau_\beta}$ versus time $t$.}
    \label{fig:numerics}
\end{figure}

\medskip
\noindent\emphbold{Numerical results. }
We further demonstrate our results through two numerical experiments.
First, we investigate the presence of QMpE in optimal thermometry using a quantum system with dimension $d=10$.
As shown in \cref{fig:numerics}(a), the state $\rho^\star$ (ground state) demonstrates faster thermalization {than} random initial states. 
Here random states are constructed as $\rho_0 = (1-\alpha)\tau_\beta + \alpha \sigma$ with $\alpha = 0.2$ and $\sigma$ being a Haar random pure state.
{These observations are consistent with \Cref{thm:prob-of-exceeding}, which predicts an exponentially small failure probability in $d$ via concentration of measure.}
Subsequently, we examine non-optimal thermometry by comparing the thermometry performance between faster-thermalized and slower-thermalized states.
In \cref{fig:numerics}(b), we present the local distinguishability $\lone{\partial_\beta \brho_{\Delta t}}$ along with the fitted convergence rates.
The convergence rate is determined by performing linear regression on $\log \fnorm{\brho_{t} - \tau_\beta}$ against time $t$ within the interval $t\in[0,1 / \gamma]$.
The initial states are Haar random pure states.
For both numerical experiments, we employ the following parameters: energy gap $\omega_1 - \omega_0 = 1$, $\gamma=1$, $\varepsilon = 0.05$, $\Delta t=0.1$, and $\beta = 1$.
{The separation between the optimal point and typical random states in \cref{fig:numerics}(b) is consistent with concentration of measure, in line with the exponential tail bound in \Cref{thm:prob-of-exceeding}.}
The optimal state consistently exhibits a larger trace norm, indicating superior performance in temperature estimation during early-stage thermalization.
It also suggests that faster-thermalizing states {tend to} demonstrate superior thermometry performance compared to slower-thermalizing states in the majority of sampled instances.
Whether this holds universally, however, remains an open question.

\medskip
\section*{Conclusion and outlook}
In this work, we establish an operational connection between QMpE and non-equilibrium quantum thermometry in a Markovian setting.
We {rigorously show} that the optimal initial states for thermometry exhibit QMpE {with a failure probability that decreases exponentially with the probe dimension}.
Our theoretical findings elucidate a {previously unexplored} relation between quantum thermodynamics and thermometry.
From a practical perspective, our results suggest that {pure probe states with anomalously rapid thermalization can be advantageous} for temperature estimation, which could serve as a criterion for selecting initial states in non-equilibrium quantum thermometry.

While fast relaxation is often viewed as harmful for metrology due to {environment-induced} information washout \cite{huelga1997ImprovementFrequencyStandardsQuantumEntanglement,zhou2018AchievingHeisenbergLimitQuantumMetrologyUsing,zhou2021AsymptoticTheoryQuantumChannelEstimation,liu2024FullyOptimizedQuantumMetrologyFrameworkToolsApplications,liu2023OptimalStrategiesQuantumMetrologyStrictHierarchy,liu2017QuantumParameterEstimationOptimalControl,li2025RecoveringOptimalPrecisionQuantumSensingTime,liu2020QuantumFisherInformationMatrixMultiparameterEstimation},
our results provide a counterpoint in a solvable Markovian thermometric setting:
{a faster} approach to equilibrium can coexist with enhanced short-time temperature distinguishability.
It would be valuable to {map the tradeoff between relaxation and distinguishability} under {relevant} constraints on interrogation time and measurements in quantum metrological tasks beyond thermometry.

Another interesting generalization is to move beyond Davies-type Markovian generators.
Non-Markovianity is known to enhance metrological {performance} in certain regimes
\cite{chin2012QuantumMetrologyNonMarkovianEnvironments,aiache2024NonMarkovianEnhancementNonequilibriumQuantumThermometry,altherr2021QuantumMetrologyNonMarkovianProcesses},
and Mpemba-like anomalies have also been reported in non-Markovian dynamics \cite{strachan2025NonMarkovianQuantumMpembaEffect}.
It remains open whether the metrological optimality--QMpE connection persists,
and if so, which dynamical features (e.g., {environmental memory} or time-dependent Liouvillians) contribute {most significantly} in metrological {tasks}.

\medskip
\section*{Methods}
This section is {organized} as follows. First, we introduce QMpE in Markovian systems, including {relevant caveats}. 
Then, {the next three subsections present} our proofs, including {(i) determining} the optimal probe state for thermometry [the proof of \cref{eq:main-text-optimal-probe-inequality-roof}], {(ii) analyzing} the spectrum of $\cL_\beta$, and {(iii) the proof of} \cref{thm:prob-of-exceeding}.
To maintain an appropriate length, we present only the key arguments and state some less central lemmas without proof. 
One can find the full proofs in the SM \cite{SM}.

\subsection*{QMpE in thermal Markovian dynamics}
We first {review} the Markovian open quantum dynamics.
The time {derivative} of {the state} of a quantum system {subject} to Markovian dynamics is {governed} by a linear map $\cL$ ({the} Liouvillian):
\begin{align}
    \frac{\d \brho}{\d t}=\cL[\brho],
\end{align}
where $\brho$ is the density matrix {of the system} and $\cL$ is defined as
\begin{align}
    \cL[\cdot] &= -i[\bH,\cdot] + \cD[\cdot],\\
    \cD[\cdot] &=\sum_j \left(L_j \cdot L_j^\dg-\frac{1}{2} \{L_j^\dg L_j, \cdot\}\right).
\end{align}
Here $\{L_j\}$ are {Lindblad} jump operators.
If {the Gibbs state is} its unique fixed point, the temperature parameter is intrinsically encoded in the dynamics.

Characterizing QMpE {for} general input states is challenging.
In general, one can analyze the thermalization behavior through {the spectral decomposition of the Liouvillian}. 
The system state after {an} evolution time $t$ can be expressed as follows \cite{Carollo2021ExponentiallyAcceleratedApproach,Moroder2024MpembaEffectThermo}
\begin{align}
    \varrho_t = e^{\cL t}[\varrho_0] = \hat r_1 + \sum_{i=2}^{d^2} e^{t \lambda_i} \Tr(\hat l_i^\dg \varrho_0) \hat r_i,
\end{align}
where $d$ is {the} system dimension, $\hat r_i$ and $\hat l_i$ are the $i$th {right} and {left} eigenoperators of $\cL$, and its {Hilbert--Schmidt adjoint} (dual map) is defined as $\cL^+[\cdot]= i [\bH,\cdot]+\sum_j\left(L_j^\dg \cdot L_j-\{L_j^\dg L_j,\cdot\}/2\right)$.
$\lambda_i$ denotes the $i$th eigenvalue.
The eigenvalues and eigenoperators can be found by solving the {eigenvalue equations}: $\cL[\hat r_i] = \lambda_i \hat r_i$ and $\cL^+[\hat l_i] = \lambda_i^* \hat l_i$, which can be {computed via} vectorization (see the SM \cite{SM}).
These eigenoperators satisfy the {biorthogonality} relation $\Tr(\hat l_i^\dg \hat r_j) = \delta_{ij}$.
We remark that this does not {uniquely} fix the norms of $\{\hat l_i\}$ and $\{\hat r_i\}$.
For convenience, we set $\fnorm{\hat r_i} = 1$.
Notice that the eigenvalues $\lambda_i$ {govern} the decay rates of different decay modes.
Therefore, we list the eigenvalues in {nondecreasing} order of the magnitude of {their real parts}, i.e., $|\Re(\lambda_1)|\le |\Re(\lambda_2)|\le...\le|\Re(\lambda_{d^2})|$, where $\lambda_1=0$ corresponds to the steady state of the system.
In the long-time limit, the {relaxation} is dominated by the slowest decaying mode, i.e., $\Re(\lambda_2)$.
If we take a state that has non-trivial overlap with $\hat l_2$ as a reference, a state can thermalize exponentially faster when it has zero overlap with $\hat l_2$.
Then it is reasonable to expect that the latter will exceed the former if they start with {appropriately chosen} initial distances.
Therefore, $\Tr(\hat l_2^\dg \varrho_0)=0$ has been considered as a {convenient indicator} of QMpE \cite{Carollo2021ExponentiallyAcceleratedApproach,Moroder2024MpembaEffectThermo}, which, strictly speaking, is neither necessary nor sufficient.
More generally, QMpE may occur even when a state has nonzero overlap with all decay modes.

{A related structural feature} is the non-normality of the Liouvillian $\cL$, which means $\cL \circ \cL^+ \neq \cL^+ \circ \cL$. 
Consequently, the {right} eigenoperators of $\cL$ are not necessarily orthogonal, i.e., $\Tr(\hat r_i^\dg \hat r_j) \neq 0$ for some $i \neq j$. 
The non-normality can be proved as follows. 
Since $\hat r_1$ represents the steady state, $\cL[\hat r_1] = 0$.
Trace preservation requires $\Tr(\cL[\hat x]) = 0$ for all {operators} $\hat x$, which is equivalent to $\Tr(\id \cL[\hat x]) = \Tr(\cL^+[\id] \hat x) = 0$. 
Hence, $\cL^+[\id] = 0$, confirming that $\id$ is a left eigenoperator with eigenvalue zero. 
Therefore, when the steady state $\hat r_1$ is {unique and not proportional to} $\id$ (i.e., not the maximally mixed state), $\cL$ is non-normal, since {normality would imply that $\id$ is also a right eigenoperator with eigenvalue zero}. 
Consequently, there exists at least one index $j \geq 2$ such that $\Tr(\hat r_j^\dg \hat r_1) \neq 0$. 
In specific instances, this constitutes the only non-orthogonal pair of eigenoperators, with orthogonality preserved among all others, i.e., $\Tr(\hat r_i^\dg \hat r_k) = 0$ for $i \neq k$ and $\{i,k\} \neq \{1,j\}$. 
Under this condition, the squared Frobenius norm of the deviation from {the} steady state simplifies to:
\begin{align}\label{eq:convergence-bound-orthogonal-case}
    \fnorm{\varrho_t - \hat r_1}^2 &= \sum_{i=2}^{d^2} \left| \Tr(\hat l_i^\dg \varrho_0) \right|^2 e^{2\Re(\lambda_i) t}.
\end{align}
Otherwise, cross terms arising from non-orthogonal eigenoperator pairs contribute to the distance, complicating the analysis of the thermalization behavior \cite{longhi2025QuantumMpembaEffectNonNormalDynamics}. {Some recent studies exploit Davies-type assumptions under which orthogonality is guaranteed within the relevant subspaces} \cite{Moroder2024MpembaEffectThermo}.

\subsection*{Optimal probe state for thermometry}
In this subsection and the following ones, we write states in a vectorized representation.
Specifically, the vectorization {in Dirac notation} is defined as
\begin{align}
    |a\rangle\langle b| \xrightarrow{\textrm{vectorize}} \ket{a}\ket{b^*}.
\end{align}
For convenience, we define $\kket{a,b}:=\ket{a}\ket{b^*}$ as {the vectorized form of the operator} $|a\rangle\langle b|$. Accordingly, the linear map is converted into matrices as well, i.e., $\mathcal{A}[\ketbra{a}{b}]=\mb{A}\kket{a,b}$.

Before presenting the proof, we first introduce the following lemma:
\begin{lemma}\label{lem:trace-norm-bound-2-main-text}
For $\mb{A},\mb{B} \in \mathbb{C}^{d\times d}, \ub{c}\in\mathbb{C}^d$ and {$a,b\in \mathbb{C}$}, {with $\alpha\in[0,1]$}, the following inequality holds:
\begin{align}
  &\left\Vert 
    \begin{pmatrix} 
      \alpha \mb{A} + (1-\alpha)\mb{B} & \sqrt{\alpha(1-\alpha)} \ub{c} \\ 
      \sqrt{\alpha(1-\alpha)} \ub{c}^{\dg} & \alpha a +(1-\alpha) b 
    \end{pmatrix}
  \right\Vert_{1} 
   \nonumber\\
  &\quad\quad\quad\quad\quad\leq\max \left\{ \left\Vert \begin{pmatrix} \mb{A} & 0 \\ 0 & a \end{pmatrix} \right\Vert_{1}, \left\Vert \begin{pmatrix} \mb{B} & 0 \\ 0 & b \end{pmatrix} \right\Vert_{1} \right\}
\end{align}
if the following conditions are satisfied:
\begin{enumerate}
  \item $|a-b| \geq 2 \left\Vert \ub{c} \right\Vert_{2}$ or $\frac{2 \left\Vert \ub{c} \right\Vert_{2}^{2}+(a-b)b}{4 \left\Vert \ub{c} \right\Vert_{2}^{2} -(a-b)^{2}} \notin [0,1]$.
  \item $\left( \left\Vert \mb{A} \right\Vert_{1} - \left\Vert \mb{B} \right\Vert_{1} \right)\left(|a|-|b|\right) \geq 0$.
\end{enumerate}
\end{lemma}
The proof of the above lemma can be found in the SM \cite{SM}.
In order to determine the optimal probe state, we need to show that the trace norm of $\partial_\beta \cL[\rho]$ is upper bounded as shown in \cref{eq:main-text-optimal-probe-inequality-roof}.
The bound can be {saturated by choosing} $\rho$ to be the ground state.
We can prove this {by applying} \cref{lem:trace-norm-bound-2-main-text}.
\begin{proof}[Proof of \cref{eq:main-text-optimal-probe-inequality-roof}.]
The {Lindblad generator} is given by
\begin{align}
    \mb{L} = \sum_{j=1}^{d-1} g_j\left[ (1+\bar n_j) \bD^{(j\rightarrow 0)} + \bar n_j \bD^{(0\rightarrow j)}\right],
\end{align}
where $\bar n_{j} = 1/[\exp(\Delta_j \beta) - 1]$, $g_j \ge 0$, and $\Gamma_\beta(\Delta_j) = g_j \bar n_j$. Let $\dbn_j:=\partial_\beta \bar n_j$, then the derivative of $\bL$ reads
\begin{align}
    \partial_\beta \bL = \sum_{k=1}^{d-1} g_k\dbn_k\boldsymbol{\Upsilon}^{(k)},
\end{align}
with $\boldsymbol{\Upsilon}^{(k)} = \bD^{(0\rightarrow k)}+ \bD^{(k\rightarrow 0)}$. Define $\til{N}=\sum_{j=1}^{d-1} g_j\bar n_j$, and {denote its $\beta$-derivative by} $\dot{\til{N}}=\sum_{j=1}^{d-1} g_j\dbn_j$.

For a {normalized} state $\ket{\psi} = \sum_{j>0} \psi_j \ket{j}$, we calculate key quantities:
\begin{align}
    a &= \bbra{0,0}\partial_\beta \bL\kket{\psi,\psi} = \sum_{k=1}^{d-1} g_k\dbn_k |\psi_k|^2,\\
    b &= \bbra{0,0} \partial_\beta \bL \kket{0,0} = -\dot{\til N},\\
    \ltwo{\ub{c}} &= \frac{1}{2} \sqrt{\sum_{k=1}^{d-1} |\psi_k|^2 \left(g_k\dbn_k+\dot{\til{N}}\right)^2}.
\end{align}
We verify the conditions of \cref{lem:trace-norm-bound-2-main-text}:
(i) $|b| \ge |a|$ follows from {$\dbn_k<0$ for all $k$ and $\sum_{k=1}^{d-1}|\psi_k|^2=1$}, which imply {$|a|=-a\le -\sum_{k=1}^{d-1} g_k\dbn_k=|b|$}.
(ii) $(a-b)^2 \le 4\ltwo{\ub{c}}^2$: {By Jensen's inequality (equivalently, Cauchy--Schwarz)}, 
\begin{align}
    (a-b)^2 - 4\ltwo{\ub{c}}^2 &= {\left(\sum_{k=1}^{d-1} |\psi_k|^2 \left(g_k \dbn_k + \dot{\til{N}}\right)\right)^2} \nonumber\\
    &\qquad- {\sum_{k=1}^{d-1} |\psi_k|^2 \left(g_k \dbn_k + \dot{\til{N}}\right)^2} \\
    &\le 0.
\end{align}
(iii) $b(a-b)+2\ltwo{\ub{c}}^2 \le 0$ holds by direct expansion, which gives {$b(a-b)+2\ltwo{\ub{c}}^2=\frac{1}{2}\left(\sum_{k=1}^{d-1}|\psi_k|^2 (g_k\dbn_k)^2-\left(\sum_{k=1}^{d-1} g_k\dbn_k\right)^2\right)\le 0$} {since $g_k\dbn_k\le 0$ for all $k$}.
(iv) For $\mb{A}=-\frac{1}{2}(\ketbra{\tilde\phi}{\psi}+\ketbra{\psi}{\tilde\phi})$ ($\ket{\tilde\phi}=\sum {g_k}\dbn_k \psi_k \ket{k}$) and $\mb{B}=\sum g_k \dbn_k \ketbra{k}{k}$, we have $\lone{\mb{B}}={-\Tr(\mb{B})}=|\dot{\til{N}}|$ {(since $\mb{B}\le 0$)} and $\lone{\mb{A}}=\sqrt{\sum g_k^2\dbn_k^2 |\psi_k|^2} \le |\dot{\til{N}}|=\lone{\mb{B}}$.

By \cref{lem:trace-norm-bound-2-main-text}, we get
\begin{align}
    \lone{\partial_\beta \cL [\rho]} \le \lone{\mb{B}}+|b|=2|\dot{\til{N}}|=2\left|\sum_{j=1}^{d-1}\frac{\partial}{\partial \beta}\Gamma_\beta(\Delta_j)\right|.
\end{align}
Saturation by the ground state is {immediate}.
\end{proof}

\medskip

\subsection*{Eigendecomposition of $\cL_\beta$ at $\varepsilon =0$}
We give analytical solutions for the degenerate case ($\varepsilon=0$).
For conciseness, we let $\Delta$ be the energy gap, $\Gamma_\beta(\Delta) = \gamma \bar n$, and $\Gamma_\beta(-\Delta) = \gamma (1+\bar n)$ where $\bar n$ is the average boson excitation number $\bar n:= [\exp(\Delta\beta)-1]^{-1}$.
According to previous analysis, one can choose $\Delta = \omega_{\rm opt}$ to maximize the thermal sensitivity.

We first consider the coherence part.
As the map is diagonal in the coherence part, we have, in {the vectorized} representation, for $k\neq i, k>0,i>0$,
\begin{align}
    \bL^\dg \kket{k,i} = -\gamma(1+\bar n) \kket{k,i},
\end{align}
and
\begin{align}
    \bL^\dg \kket{k, 0} = -\frac{1}{2}[2\gamma + \gamma \bar n (d+1)] \kket{k,0}.
\end{align}
We consider the population space spanned by $\{\kket{0},\kket{k,k}, k\in[1,d-1]\}$.
The map can be written as
\begin{align}
    L\kket{k,k} &= \gamma (1+\bar n) \kket{0,0} - \gamma (1+\bar n) \kket{k,k},\\
    L\kket{0,0} &= -\gamma(d-1) \bar n \kket{0,0} + \gamma \bar n \sum_{j=1}^{d-1} \kket{j,j}.
\end{align}
We need to compute the eigenvectors and eigenvalues of the map $L^\dg$:
\begin{align}
    L^\dg \kket{0,0} &= \gamma(1+\bar n) \sum_{j=1}^{d-1} \kket{j,j} - \gamma (d-1) \bar n \kket{0,0},\\
    L^\dg \kket{k,k} &= -\gamma (1+\bar n) \kket{k,k} + \gamma \bar n \kket{0,0}.
\end{align}
The operator $ L^\dagger $ acts on a $ d $-dimensional basis $\{ \kket{0,0}, \kket{1,1}, \dots, \kket{d-1,d-1} \}$. 
Define the symmetric state $\kket{S} = \frac{1}{\sqrt{d-1}} \sum_{j=1}^{d-1} \kket{j,j}$. The action of $ L^\dagger $ is:
\begin{align}
    L^\dagger \kket{0,0} &= -\gamma (d-1) \bar{n} \kket{0,0} + \gamma (1 + \bar{n}) \sqrt{d-1} \kket{S},\\
    L^\dagger \kket{S} &= \gamma \bar{n} \sqrt{d-1} \kket{0,0} - \gamma (1 + \bar{n}) \kket{S}.
\end{align}
The restricted $ 2\times 2 $ matrix in $\{ \kket{0,0}, \kket{S} \}$ is:
$$
M = \begin{pmatrix}
-\gamma (d-1) \bar{n} & \gamma \bar{n} \sqrt{d-1} \\
\gamma (1 + \bar{n}) \sqrt{d-1} & -\gamma (1 + \bar{n})
\end{pmatrix}.
$$
Solve $\det(M - \lambda I) = 0$:
$$
\begin{vmatrix}
-\gamma (d-1) \bar{n} - \lambda & \gamma \bar{n} \sqrt{d-1} \\
\gamma (1 + \bar{n}) \sqrt{d-1} & -\gamma (1 + \bar{n}) - \lambda
\end{vmatrix} = 0.
$$
We obtain the following eigenvalues:
$$
\lambda_1 = 0, \quad \lambda_d = -\gamma(d\bar{n} + 1).
$$
For $\lambda_1=0$, the corresponding eigenvector (up to normalization) is
\begin{align}
    \kket{l_1} = \kket{0,0} + \sum_{j=1}^{d-1}\kket{j,j} = \kket{\id}.
\end{align}
For $\lambda_d=-\gamma (d\bar n +1)$, we have
\begin{align}
    \kket{l_d} = \bar{n}\sqrt{d-1}\kket{0,0} - (1+\bar n)\kket{S}.
\end{align}

The $(d-2)$ eigenvectors in the orthogonal subspace to $\kket{0,0}$ and $\kket{S}$:
$$
\kket{\psi} = \sum_{j=1}^{d-1} c_j \kket{j,j}, \quad \sum_{j=1}^{d-1} c_j = 0, \quad L^\dagger \kket{\psi} = -\gamma(1 + \bar{n}) \kket{\psi}.
$$
Eigenvalue $\lambda_k = -\gamma(1 + \bar{n})$ with multiplicity $d-2$. Basis:
$$
\kket{\psi_k} = \frac{ \kket{1,1} - \kket{k,k} }{ \sqrt{2} }, \quad k = 2, 3, \dots, d-1.
$$

All the eigenvalues and eigenvectors in the population space are summarized in \Cref{tab:eigen-population-main}.
\begin{table}
\begin{tabular}{|c|c|c|}
\hline
\text{Eigenvalue} & \text{Multiplicity} & \text{Eigenvector} \\
\hline
$0$ & $1$ & 
    $\dfrac{ \kket{0,0} + \displaystyle\sum_{j=1}^{d-1} \kket{j,j} }{ \sqrt{d} }$ \\
\hline
$-\gamma(d \bar{n} + 1)$ & $1$ & 
    $\dfrac{ -\bar{n} \sqrt{d-1} \; \kket{0,0} + \dfrac{1 + \bar{n}}{\sqrt{d-1}} \displaystyle\sum_{j=1}^{d-1} \kket{j,j} }{ \sqrt{ \dfrac{(1 + \bar{n})^2}{d-1} + (d-1)\bar{n}^2 } }$ \\
\hline
$-\gamma(1 + \bar{n})$ & $d-2$ & 
    $\dfrac{ \kket{1,1} - \kket{k,k} }{ \sqrt{2} }$ \quad for $k = 2, 3, \dots, d-1$ \\
\hline
\end{tabular}
\caption{Eigenvectors and eigenvalues of the linear map $L$ in population space.}
\label{tab:eigen-population-main}
\end{table}
Combining with the previous results on the coherence part, we have the following sequence of inequalities:
\begin{align}
    0=|\Re(\lambda_1)| &\le |\Re(\lambda_2)|=|\Re(\lambda_3)| \cdots  \nonumber\\
    &= |\Re(\lambda_{d-1})|\le |\Re(\lambda_c)| \le |\Re(\lambda_d)|.
\end{align}
The slowest decay rate is $\lambda_2 = -\gamma(1+\bar n)$. The corresponding eigenvectors are
\begin{align}
    \frac{ \kket{1,1} - \kket{k,k} }{ \sqrt{2} } \quad \text{ for } k = 2, 3, \dots, d-1.
\end{align}
To be orthogonal to the slowest decay mode, the initial state should be either the ground state or the equal superposition of {excited} states.
The fastest decaying eigenvector can be written as (up to normalization):
\begin{align}
    -\bar{n} \sqrt{d-1} \; \kket{0,0} + \dfrac{1 + \bar{n}}{\sqrt{d-1}} \displaystyle\sum_{j=1}^{d-1} \kket{j,j},
\end{align}
which {converges} to the ground state for large $d$.

\begin{figure}
    \centering
    \includegraphics[width=0.5\linewidth]{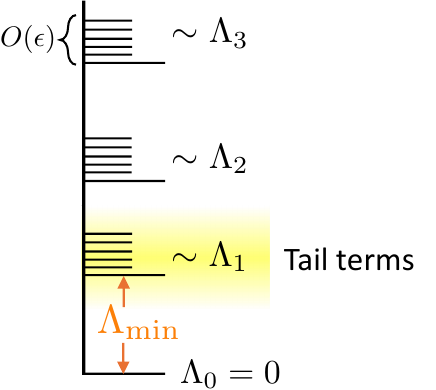}
    \caption{Spectrum of $\cL_\beta$ distributed near the unperturbed {decay rates}: $\Lambda_1 =\gamma(1+\bar n)$, $\Lambda_2 =[2\gamma + \gamma \bar n (d+1)]/2$, and $\Lambda_3 =\gamma (d\bar n + 1)$. The minimum gap $\Lambda_{\min}$ is $\Lambda_1$ {up to corrections of} $O(\varepsilon)$.}
    \label{fig:L-spec}
\end{figure}

\subsection*{Proof of \cref{thm:prob-of-exceeding}}
As we mentioned before, if the initial state has zero overlap with the lowest decay modes, it thermalizes faster than those having non-zero overlaps.
In our case, the zero-overlap condition is not satisfied {exactly}.
Instead, we have the following upper bound for the optimal probe state $\rho^\star$:
\begin{lemma}\label{thm:bound-perturbed-overlap-main}
    The optimal probe state for temperature estimation satisfies the following condition, 
    \begin{align}
        \left|\Tr(l^\dg_j \brho^{\star})\right| \le \frac{\varepsilon}{\sqrt{d-1}} \left|\frac{\partial}{\partial\omega}\log \left[e^{\omega \beta} \Gamma_\beta(\omega)\right]\right|_{\omega=\Delta},
    \end{align}
    for all {indices} $j$ that {satisfy} $|\Re(\lambda_j)|\in (0,(d-1)\Lambda_{\min}/2]$. 
    Here $\Delta:=\omega_1 - \omega_0$.
    If the excited states are exactly degenerate, the {condition is satisfied exactly},
    \begin{align}
        \Tr(l^\dg_j \brho^{\star}) = 0.
    \end{align}
\end{lemma}
We leave the proof of \cref{thm:bound-perturbed-overlap-main} in the SM \cite{SM} for conciseness. 
Based on this upper bound, we are able to {bound the contributions that decay} with rate at the slowest end of the spectrum, which we define as the tail terms (see \cref{fig:L-spec}).
Specifically, after a reasonably long time $t$, the dominant correction to the thermal state comes from the tail terms.
The upper bound in \cref{thm:bound-perturbed-overlap-main} allows us to bound the tail terms and obtain the following.

\begin{lemma}\label{lem:thermal-convergence-bound-main}
    There {exists a constant $t'>0$} such that the following bound holds for $t\ge t'$:
    \begin{align}
        &\fnorm{\rho_t^\star - \tau_\beta}^2 \le\nonumber\\
        &\quad \frac{11}{10} \varepsilon^2 e^{-2 \Lambda_{\min} t}\left(\frac{d-2}{d-1}\right)\left|\frac{\partial}{\partial\omega}\log\left[e^{\omega\beta}\Gamma_\beta(\omega)\right]\right|_{\omega = \Delta}^2,
    \end{align}
    where $\rho_t^\star:=\exp(\cL t)[\rho^\star]$.
\end{lemma}
The proof of \cref{lem:thermal-convergence-bound-main} can be found in the SM \cite{SM}.
This allows us to prove our main result.

\begin{proof}[Proof of \cref{thm:prob-of-exceeding}.]
Fix $U$ and write $\sigma=\ket{\psi}\!\bra{\psi}$ with $\ket{\psi}=U\ket{0}$.
By Lemma~\ref{lem:thermal-convergence-bound-main}, there exist $t'>0$ such that for all $t\ge t'$,
\begin{align}
\fnorm{\rho_t^\star-\tau_\beta}^2
\;\le\;
\frac{11}{10}\,\varepsilon^2\,e^{-2\Lambda_{\min}t}\,
\frac{d-2}{d-1}\,g^2,
\label{eq:opt-upper-main}
\end{align}
where $g:=\partial_\omega \log[e^{\omega\beta}\Gamma_\beta(\omega)]|_{\omega=\Delta}$.
On the other hand, since $\rho_t=(1-\alpha)\tau_\beta+\alpha e^{\cL t}[\sigma]$,
\[
\fnorm{\rho_t-\tau_\beta}^2=\alpha^2\fnorm{e^{\cL t}[\sigma]-\tau_\beta}^2.
\]
Because $\cL$ is Davies and the coherence subspace is normal (hence admits an orthogonal eigen-operator decomposition), the squared Frobenius distance contains a nonnegative contribution from the coherence modes that (in the unperturbed/near-degenerate setting) sit in the lower tail of $\mathrm{Spec}_{\cC}(\cL)$.
In particular, for all $t\ge 0$,
\begin{align}
\fnorm{\rho_t-\tau_\beta}^2
\;\ge\;
\alpha^2 e^{-2\Lambda_{\min}t}\, f(\psi),
\label{eq:rand-lower-main}
\end{align}
where
\begin{align}
f(\psi)
\;:=\;
\sum_{\substack{i,j=1\\i\neq j}}^{d-1}
\left|\bra{i}\sigma\ket{j}\right|^2
=
\sum_{\substack{i,j=1\\i\neq j}}^{d-1}
|\psi_i|^2|\psi_j|^2
\in[0,1].
\label{eq:def-f-main}
\end{align}
Combining~\eqref{eq:opt-upper-main} and~\eqref{eq:rand-lower-main}, we see that the event
\begin{align}
f(\psi)\;>\;\theta
\qquad\text{where}\qquad
\theta:=\frac{M\varepsilon^2}{\alpha^2}\frac{d-2}{d-1}g^2,
\label{eq:good-event-main}
\end{align}
implies that for all $t\ge t'$,
\begin{align}
&\fnorm{\rho_t-\tau_\beta}^2
\;\ge\;
\alpha^2 e^{-2\Lambda_{\cL}^{\min}t} f(\psi)\nonumber\\
&\qquad\;>\;
M\varepsilon^2 e^{-2\Lambda_{\cL}^{\min}t}\frac{d-2}{d-1}g^2
\;\ge\;
\fnorm{\rho_t^\star-\tau_\beta}^2,
\end{align}
i.e., $\rho_t^\star$ exceeds $\rho_t$.
Therefore,
\begin{align}
\delta
\;\le\;
\Pr_{U\sim\mu_H}\!\big[f(\psi)\le \theta\big].
\label{eq:delta-reduction-main}
\end{align}

\smallskip\noindent
For Haar-random $\ket{\psi}\in\mathbb{C}^d$, one has the standard moment identity
\[
\mathbb{E}\big[|\psi_i|^2|\psi_j|^2\big]=\frac{1}{d(d+1)}
\qquad (i\neq j).
\]
Hence, by~\eqref{eq:def-f-main},
\begin{align}
\mathbb{E}_{U\sim\mu_H} f(\psi)
=
\sum_{\substack{i,j=1\\i\neq j}}^{d-1}\frac{1}{d(d+1)}
=
\frac{(d-1)(d-2)}{d(d+1)}
=
\mu_d.
\label{eq:Ef-main}
\end{align}

Next, view $f$ as a function on the unit sphere $S^{2d-1}$ (identify $\mathbb{C}^d\simeq\mathbb{R}^{2d}$).
Let $\ket{\psi},\ket{\phi}$ be unit vectors and write $\sigma_\psi=\ket{\psi}\!\bra{\psi}$, $\sigma_\phi=\ket{\phi}\!\bra{\phi}$.
Let $\mathcal{P}_\cC$ denote the orthogonal projector (w.r.t.\ the Hilbert--Schmidt inner product) onto the operator subspace
$\mathrm{span}\{\ket{i}\!\bra{j}: i,j\in\{1,\dots,d-1\},\, i\neq j\}$.
Then $f(\psi)=\|\mathcal{P}_\cC(\sigma_\psi)\|_2^2$.
Therefore,
\begin{align*}
|f(\psi)-f(\phi)|
&=
\Big|
\|\mathcal{P}_\cC(\sigma_\psi)\|_2^2-\|\mathcal{P}_\cC(\sigma_\phi)\|_2^2
\Big|\\
&\le
\big(\|\mathcal{P}_\cC(\sigma_\psi)\|_2+\|\mathcal{P}_\cC(\sigma_\phi)\|_2\big)\,
\|\mathcal{P}_\cC(\sigma_\psi-\sigma_\phi)\|_2\\
&\le 2\,\|\sigma_\psi-\sigma_\phi\|_2.
\end{align*}
Using $\|\sigma_\psi-\sigma_\phi\|_2^2=2-2|\langle\psi|\phi\rangle|^2\le 2\|\psi-\phi\|_2^2$, we get
\[
|f(\psi)-f(\phi)|\le 2\sqrt{2}\,\|\psi-\phi\|_2.
\]
Thus $f$ is $L$-Lipschitz on $S^{2d-1}$ with $L=2\sqrt{2}$.

\smallskip\noindent
A standard form of Lévy's lemma (e.g.\ \cite{mele2024IntroductionHaarMeasureToolsQuantumInformation}) states that for an $L$-Lipschitz function $h:S^{n}\to\mathbb{R}$,
\[
\Pr\big(|h-\mathbb{E}h|\ge a\big)\le 2\exp\!\left(-\frac{(n+1)a^2}{9\pi^3 L^2}\right).
\]
Apply this with $h=f$, $n=2d-1$, $L=2\sqrt{2}$, and $a:=\mu_d-\theta$.
When $a>0$, we have
\begin{align}
\Pr[f(\psi)\le \theta]
&\le
\Pr\big(|f(\psi)-\mu_d|\ge \mu_d-\theta\big)\\
&\le
2\exp\!\left(-\frac{2d\,(\mu_d-\theta)^2}{9\pi^3\cdot 8}\right)\\
&=
2\exp\!\left(-\frac{d}{36\pi^3}(\mu_d-\theta)^2\right).
\end{align}
Combining with~\eqref{eq:delta-reduction-main} and substituting~\eqref{eq:good-event-main} completes the proof and yields~\eqref{eq:delta-exp-bound-main}.
\end{proof}

\medskip
\section*{DATA AVAILABILITY}
Data generated and analyzed during current study are available from the corresponding author upon reasonable request.

\medskip
\section*{CODE AVAILABILILTY}
Code used to generate data in this study are available from the corresponding author upon reasonable request.

\medskip
\section*{Acknowledgement}
This work is supported by the National Natural Science Foundation of China via the Excellent Young Scientists Fund (Hong Kong and Macau) Project 12322516, the National Natural Science Foundation of China (NSFC)/Research Grants Council (RGC) Joint Research Scheme via Project N\_HKU7107/24, Guangdong Provincial Quantum Science Strategic Initiative via Projects GDZX2403008 and GDZX2503001,  and the Hong Kong Research Grant Council (RGC) through grant 17302724.

\medskip
\emph{Note added.---}After completing this work, we became aware of an independent work \cite{chattopadhyay2026AnomalyResourceMpembaEffectQuantumThermometry} about the QMpE and quantum thermometry. While Ref. \cite{chattopadhyay2026AnomalyResourceMpembaEffectQuantumThermometry} considers the finite-time regime, we focus on short-time thermometric optimality within $d$-dimensional Markovian dynamics, establishing a rigorous connection to the QMpE.

\medskip
\section*{AUTHOR CONTRIBUTIONS}
Both authors contributed to the design and the implementation of the research as well as the writing of the manuscript.

\medskip
\section*{COMPETING INTERESTS}
The authors declare no competing interests.

\bibliography{ref.bib}
\clearpage
\begin{widetext}
\begin{center}
    {\Large\bfseries Supplemental Material}
\end{center}

\tableofcontents

\setcounter{section}{0}
\let\addcontentsline\oldaddcontentsline
\section{Preliminaries}\label{app:preliminaries}
\noindent\emph{Vectorization representation.}
We frequently employ matrix vectorization in our calculations. A brief introduction is provided here (see Ref. \cite{watrous2018TheoryQuantumInformation} Sec. 1.1.2 for a comprehensive review).
This formulation provides an alternative representation of a matrix; for instance, a matrix $\mb{M}\in \mathbb{C}^{2\times 2}${ is vectorized as follows:}
\begin{align}
    \mb{M} = \begin{pmatrix}
        a & b\\
        c & d
    \end{pmatrix} \xrightarrow{\textrm{vectorize}} \kket{\mb M} = \begin{pmatrix}
        a\\
        b\\
        c\\
        d
    \end{pmatrix}.
\end{align}
{Using Dirac notation, this vectorization is formally defined as}
\begin{align}
    |a\rangle\langle b| \xrightarrow{\textrm{vectorize}} \ket{a}\ket{b^*}.
\end{align}
{For brevity, we define }$\kket{a,b}:=\ket{a}\ket{b^*}${ as the vectorized form of the operator }$|a\rangle\langle b|${.}
{We summarize several fundamental properties of this representation:}
\begin{align}
    (\mb{A} \ot \mb{B}) \kket{\ub C} &= \kket{\mb{A} \ub C \mb{B}^T},\label{eq:property-vec}\\
    \bbkk{\mb{A}}{\mb{B}} &= \Tr(\mb{A}^\dg \mb{B}).
\end{align}

\noindent\emph{Operator eigen-decomposition.}
{For a linear map }$\mathcal{A}: \mathbb{C} ^{d\times d} \mapsto \mathbb{C}^{d\times d}${, we define its corresponding matrix representation as }$\mb{A} \in \mathbb{C}^{d^2 \times d^2}${, satisfying}
\begin{align}
    \kket{\mathcal{A}[X]} = \mb{A} \kket{X}.
\end{align}
{The eigen-decomposition of }$\mathcal{A}${ is derived directly from that of }$\mb{A}${, given by }$\mb{A} = \sum_{i=1}^{d^2} a_i \kkbb{\hat r_i}{\hat l_i}${, where }$a_i${ denotes the }$i${-th eigenvalue, while }$\kket{\hat r_i}${ and }$\kket{\hat l_i}${ represent the corresponding right and left eigenvectors, respectively.}
{These eigenvectors satisfy the biorthogonality condition }$\bbkk{\hat l_i}{\hat r_j}=\Tr(\hat l_i^\dg \hat r_j) = \delta_{ij}${.}
{Consequently, the action of }$\mathcal{A}${ on an arbitrary operator }$X${ is expressed as}
\begin{align}
    \mathcal{A}[X] = \sum_{i=1}^{d^2} a_i \Tr(\hat l_i^\dg X) \hat r_i.
\end{align}
{We note that the constraints on }$\{\hat l_i, \hat r_i\}${ are insufficient to uniquely determine their norms.} 
{Specifically, a scaling degree of freedom remains, allowing the transformations }$\hat l_i \rightarrow \alpha_i \hat l_i${ and }$\hat r_i \rightarrow \hat r_i/\alpha_i${.} 
{To eliminate this ambiguity, we adopt the convention }$\fnorm{\hat r_i} = 1${.}
{By applying the Cauchy-Schwarz inequality, we obtain}
\begin{align}
    \fnorm{\hat l_i}^2=\fnorm{\hat l_i}^2 \fnorm{\hat r_i}^2 \ge |\Tr(\hat l_i ^\dg \hat r_i)|^2 = 1.
\end{align}

\section{Matrix inequalities of trace norm}
{In this section, we establish several matrix inequalities pertaining to the trace norm, which will be utilized in subsequent proofs.}
\begin{lemma}\label{lem:trace-norm-bound}
{For any block matrix constructed from }$\mb{A} \in \mathbb{C}^{d\times d}, \ub{c}\in\mathbb{C}^d${ and }$b\in \mathbb{C}${, the following inequality holds:}
\[
\lone{\mat{\mb{A}}{\ub{c}}{\ub{c}^\dg}{b}} \leq \lone{\mb{A}} + \sqrt{4 \ltwo{\ub{c}}^2 + |b|^2}, 
\]
{where }$\mb{A} \in \mathbb{C}^{d\times d}${, }$\ub{c}\in\mathbb{C}^d${, and }$b\in \mathbb{C}${.}
\end{lemma}

\begin{proof}
{By applying the variational characterization of the trace norm, we have}
\begin{align*}
\lone{\mat{\mb{A}}{\ub{c}}{\ub{c}^\dg}{b}} &= \sup_{\|M\|_2 \leq 1} \operatorname{Tr}\left[M \mat{\mb{A}}{\ub{c}}{\ub{c}^\dg}{b}\right] \\
&= \sup_{\mb{D}, \ub{e}, f} \operatorname{Tr}\left[\mat{\mb{D}}{\ub{e}}{\ub{e}^\dg}{f} \mat{\mb{A}}{\ub{c}}{\ub{c}^\dg}{b}\right] \\
&= \sup_{\mb{D}, \ub{e}, f} \operatorname{Tr}\left[\mat{\mb{D}\mb{A} + \ub{e}\ub{c}^\dg}{*}{*}{\ub{e}^\dg\ub{c} + fb}\right] \\
&= \sup_{\mb{D}, \ub{e}, f} \operatorname{Tr}[\mb{D}\mb{A} + \ub{e}\ub{c}^\dg + \ub{e}^\dg\ub{c} + fb] \\
&= \sup_{\mb{D}, \ub{e}, f} \operatorname{Tr}(\mb{D}\mb{A}) + fb + 2\operatorname{Re}[\ub{c}^\dg\ub{e}].
\end{align*}
{where the operator }$\mb{D}${, the vector }$\ub{e}${, and the scalar }$f${ constitute a block matrix }$M${ satisfying }$\|M\|_2 \leq 1${.} 
{Selecting a test vector }$\ub{x} = \begin{pmatrix} \ub{x}_0 \\ 0 \end{pmatrix}${ with }$\|\ub{x}_0\|_2^2=1${, we obtain:}
\[
1 \geq \sup_{\ub{x}_0} \left\| \mat{\mb{D}}{\ub{e}}{\ub{e}^\dg}{f} \begin{pmatrix} \ub{x}_0 \\ 0 \end{pmatrix} \right\|_2^2 
= \sup_{\ub{x}_0} \left\| \begin{pmatrix} \mb{D}\ub{x}_0 \\ \ub{e}^\dg\ub{x}_0 \end{pmatrix} \right\|_2^2 
= \sup_{\ub{x}_0} \left( \|\mb{D}\ub{x}_0\|_2^2 + |\ub{e}^\dg\ub{x}_0|^2 \right).
\]
{This yields the necessary condition:}
\begin{equation}
\|\mb{D}\|_2^2 \leq 1. \label{eq:relax-cond-1}
\end{equation}
{Similarly, by choosing }$\ub{x} = \begin{pmatrix} 0 \\ \vdots \\ 1 \end{pmatrix}${, we deduce:}
\begin{equation}
|f|^2 + \|\ub{e}\|_2^2 \leq 1. \label{eq:relax-cond-2}
\end{equation}
{Furthermore, the real part of the inner product is bounded by the Cauchy-Schwarz inequality as:}
\begin{equation}
\operatorname{Re}(\ub{c}^\dg\ub{e}) \leq |\ub{c}^\dg\ub{e}| \leq \ltwo{\ub{c}} \|\ub{e}\|_2. \label{eq:inner-product-bound-cond3}
\end{equation}
{Combining these inequalities yields:}
\begin{equation}
\sup_{\mb{D}, \ub{e}, f} \operatorname{Tr}(\mb{D}\mb{A}) + fb + 2\operatorname{Re}[\ub{c}^\dg\ub{e}] 
\leq \lone{\mb{A}} + \max_{|f|} \lr{2 \ltwo{\ub{c}} \sqrt{1-|f|^2} + |b| |f|}. \label{eq:trace-norm-bound-1}
\end{equation}
{By substituting }$\theta = |f|^2${, we can further bound the maximum via the Cauchy-Schwarz inequality:}
\begin{align*}
&\max_{\theta} \lr{2 \ltwo{\ub{c}} \sqrt{1-\theta} + |b|\sqrt{\theta}} \\
&= \max_{\theta} \begin{pmatrix} 2\ltwo{\ub{c}} \\ |b| \end{pmatrix} \bdot \begin{pmatrix} \sqrt{1-\theta} \\ \sqrt{\theta} \end{pmatrix} \\
&\leq \sqrt{\left(2\ltwo{\ub{c}}\right)^2 + |b|^2} \cdot \sqrt{(\sqrt{1-\theta})^2 + (\sqrt{\theta})^2} \\
&= \sqrt{4 \ltwo{\ub{c}}^2 + |b|^2},
\end{align*}
{which completes the proof.}
\end{proof}

\begin{lemma}
\label{lem:trace-norm-bound-2}
{Under specific spectral conditions, the following trace norm inequality holds:}
\[
  \left\Vert 
    \begin{pmatrix} 
      \alpha \mb{A} + (1-\alpha)\mb{B} & \sqrt{\alpha(1-\alpha)} \ub{c} \\ 
      \sqrt{\alpha(1-\alpha)} \ub{c}^{\dg} & \alpha a +(1-\alpha) b 
    \end{pmatrix}
  \right\Vert_{1} 
  \leq \max \left\{ \left\Vert \begin{pmatrix} \mb{A} & 0 \\ 0 & a \end{pmatrix} \right\Vert_{1}, \left\Vert \begin{pmatrix} \mb{B} & 0 \\ 0 & b \end{pmatrix} \right\Vert_{1} \right\}
\]
{provided that the following conditions are met:}
\begin{enumerate}
  \item $|a-b| \geq 2 \left\Vert \ub{c} \right\Vert_{2}${ or }$\frac{2 \left\Vert \ub{c} \right\Vert_{2}^{2}+(a-b)b}{4 \left\Vert \ub{c} \right\Vert_{2}^{2} -(a-b)^{2}} \notin [0,1]$.
  \item $\left( \left\Vert \mb{A} \right\Vert_{1} - \left\Vert \mb{B} \right\Vert_{1} \right)\left(|a|-|b|\right) \geq 0$.
\end{enumerate}
\end{lemma}

\begin{proof}
{By invoking Lemma~\ref{lem:trace-norm-bound}, we establish that for a matrix of the form}
\[
  \mb{Q} = \begin{pmatrix} 
        \alpha \mb{A} + (1-\alpha)\mb{B} & \sqrt{\alpha(1-\alpha)} \ub{c} \\ 
        \sqrt{\alpha(1-\alpha)} \ub{c}^{\dg} & \alpha a +(1-\alpha) b
      \end{pmatrix},
\]
{its trace norm is bounded above by}
\[
  \left\Vert \mb{Q} \right\Vert_{1} \leq \left\Vert \alpha \mb{A} + (1-\alpha)\mb{B} \right\Vert_{1} + \sqrt{ 4\alpha(1-\alpha) \left\Vert \ub{c} \right\Vert_{2}^{2} + \left[ \alpha a +(1-\alpha) b \right]^{2} }
\]

{Defining }$g(\alpha) = 4\alpha(1-\alpha) \left\Vert \ub{c} \right\Vert_{2}^{2} + \left[ \alpha a +(1-\alpha) b \right]^{2}${, we compute its derivative as}
\begin{align*}
  \frac{\d g}{\d \alpha} &= 4(1-2\alpha) \left\Vert \ub{c} \right\Vert_{2}^{2} + 2(a-b)\left[ \alpha a +(1-\alpha) b \right] \\
  &= 4 \left\Vert \ub{c} \right\Vert_{2}^{2} - 8 \alpha \left\Vert \ub{c} \right\Vert_{2}^{2} + 2(a-b)\left[ (a-b)\alpha + b \right] \\
  &= 4 \left\Vert \ub{c} \right\Vert_{2}^{2} - 8 \alpha \left\Vert \ub{c} \right\Vert_{2}^{2} + 2(a-b)^{2} \alpha + 2(a-b)b \\
  &= \alpha \left[ 2(a-b)^{2} - 8 \left\Vert \ub{c} \right\Vert_{2}^{2} \right] + 2(a-b)b + 4 \left\Vert \ub{c} \right\Vert_{2}^{2}.
\end{align*}
{We now analyze two scenarios to determine when the maximum is achieved at the boundaries:}
\begin{itemize}
  \item {If }$(a-b)^{2} \geq 4 \left\Vert \ub{c} \right\Vert_{2}^{2}${, the function }$g${ is convex, implying that its maximum is attained at the boundaries of the feasible region, i.e., }$\alpha=0${ or }$1${.} 
  \item {If }$(a-b)^{2} < 4 \left\Vert \ub{c} \right\Vert_{2}^{2}${, the function }$g${ is strictly concave; thus, the maximum occurs at the boundary provided the critical point lies outside the feasible interval }$[0,1]${:}
\[
  \alpha^{\ast} = \frac{ -2(a-b)b - 4 \left\Vert \ub{c} \right\Vert_{2}^{2} }{ 2(a-b)^{2} - 8 \left\Vert \ub{c} \right\Vert_{2}^{2} } = \frac{ 2 \left\Vert \ub{c} \right\Vert_{2}^{2} + (a-b)b }{ 4 \left\Vert \ub{c} \right\Vert_{2}^{2} - (a-b)^{2} } \notin [0, 1].
\]
\end{itemize}
{When these conditions are met, the relevant term can be bounded as follows:}
\[
  \max_{\theta, \alpha} \left( 2 \sqrt{\alpha(1-\alpha)} \left\Vert \ub{c} \right\Vert_{2} \sqrt{1-\theta} + \left| \alpha a +(1-\alpha) b \right| \sqrt{\theta} \right) \leq \max\{ |a|, |b| \}.
\]
{To deduce the final upper bound, we require the following inequalities to hold:}
\begin{align*}
  &\left\Vert \mb{A} \right\Vert_{1} \geq \left\Vert \mb{B} \right\Vert_{1} \text{ and } |a| \geq |b| \quad \text{or} \\
  &\left\Vert \mb{A} \right\Vert_{1} < \left\Vert \mb{B} \right\Vert_{1} \text{ and } |a| < |b|,
\end{align*}
{which can be compactly rewritten as the equivalent condition}
\[
  \left( \left\Vert \mb{A} \right\Vert_{1} - \left\Vert \mb{B} \right\Vert_{1} \right)\left(|a|-|b|\right) \geq 0.
\]
{Consequently, we arrive at the final trace norm bound:}
\begin{align*}
  \left\Vert \mb{Q} \right\Vert_{1} &\leq \max \left\{ \left\Vert \mb{A} \right\Vert_{1}, \left\Vert \mb{B} \right\Vert_{1} \right\} + \max\{ |a|, |b| \} \\
  &= \max \left\{ \left\Vert \begin{pmatrix} \mb{A} & 0 \\ 0 & a \end{pmatrix} \right\Vert_{1}, \left\Vert \begin{pmatrix} \mb{B} & 0 \\ 0 & b \end{pmatrix} \right\Vert_{1} \right\},
\end{align*}
{which completes the proof.}
\end{proof}

\section{The state with optimal thermal sensitivity}\label{app:opt-init-state}

Consider $d$ dimensional system with ground state $\ket0$ and different exited states $\ket{j},j=1,2,\dots,d-1$. 
We only consider the energy exchange bewteen the ground state and the exited states.
This means that we can use the transition weight function $\Gamma_\beta(\omega)=J(\omega)[\exp(\beta\omega)-1]^{-1}$ to obtain the transition strengths $\{g_j\ge 0\}$ for all allowed jumps:
\begin{align}
    g_j&=J(\Delta_j)=-J(-\Delta_j).
\end{align}
where $\delta(\cdot)$ denotes the Dirac-$\delta$ function and $\Delta_j=E_j-E_0$.

We have the following lemma:
\begin{lemma}\label{lem:optimal-thermometry-star-shape}
    The trace norm of $\partial_\beta \cL[\rho]$ is upper-bounded as following:
    \begin{align}
        \lone{\partial_\beta \cL [\rho]} \le 2\left|\sum_{j=1}^{d-1}\frac{\partial}{\partial \beta}\Gamma_\beta(\Delta_j)\right|.
    \end{align}
    The bound can be saturated by setting $\rho$ to be the ground state.
\end{lemma}
\begin{proof}
First, we have the Lindblad operator in the following form
\begin{align}
    \mb{L} = \sum_{j=1}^{d-1} g_j\left[ (1+\bar n_j) \bD_\da^{(j)} + \bar n_j \bD_\ua^{(j)}\right],
\end{align}
where $\bar n_{j} = 1/[\exp(\Delta_j \beta) - 1]$ and $g_j \ge 0$ for all $j=1,2,...,d-1$. 
$\bar n_{j}$ is also called Bose function that satisfies the following KMS conditions
\begin{align}
    \frac{1+\bar n_j}{\bar n_j} = e^{\beta \Delta_j},~\forall j=1,2,...,d-1.
\end{align}
In addition, we have
\begin{align}
    \Gamma_\beta(\Delta_j) = g_j \bar n_j.
\end{align}
Here we define the following operators
\begin{align}
    \boldsymbol{\Upsilon}^{(k)} &= \bD_\ua^{(k)} + \bD_\da^{(k)}\\
    &=\kket{0,0}\bbra{k,k} + \kket{k,k}\bbra{0,0} -\frac{1}{2}\left(\ketbra{k}\ot \id + \id\ot\ketbra{k^*}+\ketbra{0}\ot \id + \id \ot\ketbra{0^*}\right).
\end{align}
Let $\dbn_j=\partial_\beta \bar n_j$,
such that the derivative of $\bL$ can be written as
\begin{align}
    \partial_\beta \bL = \sum_{k=1}^{d-1} g_k\dbn_k\boldsymbol{\Upsilon}^{(k)},
\end{align}
where $\dbn_k = \partial_\beta \bar{n}_k$.
Let the state $\ket{\psi} = \sum_{j>0} \psi_j \ket{j}$. We have $\kket{\psi,\psi}=\sum_{i,j} \psi_i \psi_j^* \kket{i,j}$ and
\begin{align}
    a = \bbra{0,0}\partial_\beta \bL\kket{\psi,\psi} =\sum_{k=1}^{d-1} g_k\dbn_k |\psi_k|^2.
\end{align}
Similarly,
\begin{align}
    b = \bbra{0,0} \partial_\beta \bL \kket{0,0} = - \dot{\til N},
\end{align}
where $\til{N}=\sum_{j=1}^{d-1} g_j\bar n_j$.
By computing $\partial_\beta L \kket{\psi,0}$, we have
\begin{align}
    \ltwo{\ub{c}} = \frac{1}{2} \sqrt{\sum_{k=1}^{d-1} |\psi_k|^2 \left(g_k\dbn_k+\dot{\til{N}}\right)^2}.
\end{align}
We also have
\begin{align}
    \mb{A} = -\frac{1}{2} \left(\ketbra{\tilde\phi}{\psi} + \ketbra{\psi}{\tilde{\phi}}\right),
\end{align}
with $\ket{\tilde{\phi}}=\sum_{k=1}^{d-1} \dbn_k \psi_k \ket{k}$.
And
\begin{align}
    \mb{B} = \sum_{k=1}^{d-1} g_k \dbn_k \ketbra{k}{k}.
\end{align}
We then verify the conditions of \cref{lem:trace-norm-bound-2}.
Specifically, we want to prove the following conditions hold: $(a-b)^2 \le 4\ltwo{\ub{c}}^2$, $b(a-b)+2\ltwo{\ub{c}}^2\le 0$, $\lone{A} \le \lone{B}$ and $|b|\ge |a|$.
First, it is obvious that $|b| \ge |a|$ as $\dot{\til{N}}^2 = \left(\sum_{j=1}^{d-1} g_j\dbn_j\right)^2 \ge g_k^2\dbn_k^2$ holds for $k=1,2,...,d-1$.
We check the first condition:
\begin{align}
    (a-b)^2 - 4\ltwo{\ub{c}}^2&= \left(\dot{\til N}+\sum_{k=1}^{d-1} g_k \dbn_k |\psi_k|^2\right)^2 - \sum_{l=1}^{d-1} |\psi_l|^2 \left(g_l \dbn_l + \dot{\til N}\right)^2\\
    &= \left(\sum_{k=1}^{d-1} g_k \dbn_k |\psi_k|^2\right)^2 + \dot{\til{N}}^2 + 2\dot{\til{ N}} \sum_{j=1}^{d-1}g_j \dbn_j |\psi_j|^2 - \sum_{l=1}^{d-1} |\psi_l|^2 \left(g_l^2 \dbn_l^2+\dot{\til{N}}^2+2\dot{\til{N}}g_l \dbn_l\right)\\
    &=\left(\sum_{k=1}^{d-1} g_k \dbn_k |\psi_k|^2\right)^2 - \sum_{l=1}^{d-1} |\psi_l|^2 g_l^2 \dbn_l^2\\
    &=\left(\sum_{k=1}^{d-1} g_k \dbn_k |\psi_k|^2\right)^2 - \left(\sum_{l=1}^{d-1} |\psi_l|^2 g_l^2 \dbn_l^2\right)\left(\sum_{j=1}^{d-1} |\psi_j|^2\right)\\
    &\le \left(\sum_{k=1}^{d-1} g_k \dbn_k |\psi_k|^2\right)^2 - \left(\sum_{j=1}^{d-1}|\psi_j|^2 g_j \dbn_j\right)^2=0,
\end{align}
where the final inequality comes directly from Cauchy–Schwarz inequality.
To verify the second condition, we have
\begin{align}
    b (a -b) + 2\ltwo{\ub{c}}^2 &= -\dot{\til{N}} \left( \sum_{k=1}^{d-1}g_k\dbn_k|\psi_k|^2 + \dot{\til{N}}\right) + \frac{1}{2} \sum_{k=1}^{d-1}|\psi_k|^2 \left(g_k\dbn_k + \dot{\til N}\right)^2\\
    &= -\dot{\til{N}} \left( \sum_{k=1}^{d-1}g_k\dbn_k|\psi_k|^2 + \dot{\til{N}}\right) + \frac{1}{2} \sum_{k=1}^{d-1}|\psi_k|^2 \left(g_k^2\dbn_k^2+\dot{\til{N}}^2 + 2g_k\dbn_k \dot{\til{N}}\right)\\
    &=\frac{1}{2}\sum_{k=1}^{d-1}|\psi_k|^2 \left(g_k^2\dbn_k^2-\dot{\til{N}}\right)\le 0.
\end{align}
The trace norm of $\mb{B}$ is easy to obtain as it is diagonal:
\begin{align}
    \lone{\mb{B}} = \sum_{k=1}^{d-1} |g_k\dbn_k| = |\dot{\til{N}}|.
\end{align}
To compute the trace norm of $\mb{A}$, we let $\ket{v} = a \ket{\tilde\phi}+b\ket{\psi}$. We have
\begin{align}
    \mb{A} \ket{v} = \lambda\ket{v}\Rightarrow -\frac{\gamma}{2}\left(a\braket{\psi|\tilde{\phi}}\ket{\tilde \phi}+a\braket{\tilde{\phi}|\tilde{\phi}} \ket{\psi} + b\ket{\tilde{\phi}}+b\braket{\til{\phi}|\psi}\ket{\psi}\right) = \lambda a\ket{\til \phi} +\lambda b\ket{\psi}.
\end{align}
Accordingly, we have
\begin{align}
    \frac{\gamma}{2} \left(a \braket{\psi|\til \phi} + b \right) +\lambda a&=0,\\
    \frac{\gamma}{2} \left(a \braket{\til \phi|\til \phi}+ b\braket{\til \phi|\psi}\right) + \lambda b&=0.
\end{align}
In matrix form, we have
\begin{align}
    \begin{pmatrix}
        \lambda + \frac{1}{2} \braket{\psi|\til \phi} & \frac{1}{2}\\
        \frac{1}{2} \braket{\til \phi|\til \phi} & \lambda + \frac{1}{2}\braket{\til \phi|\psi}
    \end{pmatrix}
    \begin{pmatrix}
        a\\b
    \end{pmatrix}=\ub{0}.
\end{align}
We have
\begin{align}
    \det \begin{pmatrix}
        \lambda + \frac{1}{2} \braket{\psi|\til \phi} & \frac{1}{2}\\
        \frac{1}{2} \braket{\til \phi|\til \phi} & \lambda + \frac{1}{2}\braket{\til \phi|\psi}
    \end{pmatrix} = \left(\lambda + \frac{1}{2} \braket{\psi|\til \phi}\right)\left(\lambda + \frac{1}{2}\braket{\til \phi|\psi}\right) - \frac{1}{4}\braket{\til\phi|\til\phi} = 0.
\end{align}
By solving this equation, we have
\begin{align}
    \lambda_{1,2} = \frac{1}{2}\left(-\sum_{j=1}^{d-1} |\psi_j|^2 g_j \dbn_j \pm \sqrt{\sum_{k=1}^{d-1} g_k^2\dbn_k^2 |\psi_k|^2}\right).
\end{align}
By Cauchy–Schwarz inequality, we have
\begin{align}
    \left(\sum_j |\psi_j|^2\right)\left(\sum_j |\psi_j|^2\dbn_j^2\right)\ge \left(\sum_j |\psi_j|^2 \dbn_j \right)^2.
\end{align}
Then
\begin{align}
    \lone{\mb{A}} = |\lambda_1 - \lambda_2| = \sqrt{\sum_{k=1}^{d-1} g_k^2\dbn_k^2 |\psi_k|^2}.
\end{align}
As $\dot N^2 \ge \dbn_k^2$ holds for all $k$'s, we have $\lone{A}\le \lone{B}$.
In conclusion, the local distinguishability of the output states is bounded as
\begin{align}
    \lone{\partial_\beta \cL[\brho]} \le \lone{\mb{B}}+|b|=2\left|\dot{\til{N}}\right|=2\left|\sum_{j=1}^{d-1}\frac{\partial}{\partial \beta}\Gamma_\beta(\Delta_j)\right|,
\end{align}
which is saturated by choosing the ground state as the initial state.
\end{proof}

\section{Exact solution for the degenerate case}\label{app:satisfy-SAT}
For convenience, we write the explicit form of the maps $\bD_\ua^{(j)}$ and $\bD_\da^{(j)}$ in the vectorization representation:
\begin{align}
\bD_\da^{(j)} &=\kket{0,0} \bbra{j, j} - \frac{1}{2}\left( |j\rangle\langle j|\ot \id + \id \ot |j^*\rangle\langle j^*| \right),\\
\bD_\ua^{(j)} &= \kket{j,j} \bbra{0, 0} - \frac{1}{2}\left( |0\rangle\langle 0|\ot \id + \id \ot |0^*\rangle\langle 0^*| \right).
\end{align}
Their adjoint maps can be written as
\begin{align}
\bD_\da^{(j)\dg} &= \kket{j,j} \bbra{0,0} - \frac{1}{2}\left( |j\rangle\langle j|\ot \id + \id \ot |j^*\rangle\langle j^*| \right),\\
\bD_\ua^{(j)\dg} &= \kket{0,0} \bbra{j, j} - \frac{1}{2}\left( |0\rangle\langle 0|\ot \id + \id \ot |0^*\rangle\langle 0^*| \right).
\end{align}
As $[\sum_j \mb{A}_j, \sum_k \mb{B}_k]=\sum_{j,k}[\mb{A}_j,\mb{B}_k]$,
we can check the commutator separately, i.e.,
\begin{align}
\left[\bD_\da^{(j)},\bD_\da^{(k)\dg}\right] &= \bD_\da^{(j)}\bD_\da^{(k)\dg}-\bD_\da^{(j)\dg}\bD_\da^{(k)}\\
&= \kkbb{0,0}{0,0}\delta_{j,k} - \kkbb{j,j}{0,0}\delta_{j,k}-\kkbb{0,0}{j,j}\delta_{j,k} - \kkbb{j,j}{k,k},\\
\left[\bD_\da^{(j)},\bD_\ua^{(k)\dg}\right] &= (1 + \delta_{jk}) \kket{0,0} \bbra{j,j},\\
\left[\bD_\ua^{(j)},\bD_\da^{(k)\dg}\right] &= - (1 + \delta_{jk}) \kket{k,k} \bbra{0,0},\\
\left[\bD_\ua^{(j)},\bD_\ua^{(k)\dg}\right] &= \kket{j,j} \bbra{k,k} - \kket{j,j} \bbra{0,0} - \kket{0,0} \bbra{k,k} - \delta_{jk} \kket{0,0} \bbra{0,0}.
\end{align}
In addition, we have
\begin{align}
\sum_{j=1}^{d-1} \bD_\da^{(j)\dg} &= \sum_{j=1}^{d-1} \kket{j,j}\bbra{0,0} - \id\ot\id,\\
\sum_{j=1}^{d-1} \bD_\ua^{(j)\dg} &= \sum_{j=1}^{d-1} \kket{0,0}\bbra{j,j} - \frac{d-1}{2}\left( |0\rangle\langle 0|\ot \id + \id \ot |0^*\rangle\langle 0^*| \right).
\end{align}
We define $\bSigma_\da:=\sum_{j=1}^{d-1} \bD_\da^{(j)}$ and $\bSigma_\ua:=\sum_{j=1}^{d-1} \bD_\ua^{(j)}$.
We now consider the action of the above maps on the population space spanned by $\{\kket{0,0}, \kket{j,j}, j=1,2,...,d-1\}$.
\begin{align}
\bSigma^\dg_\da\kket{0,0} &= - \kket{0,0} + \sum_{j=1}^{d-1} \kket{j,j},\\
\bSigma^\dg_\da\kket{k,k} &= -\kket{k,k},~\forall k=1,2,...,d-1,\\
\bSigma^\dg_\ua\kket{0,0} &= -(d-1)\kket{0,0},\\
\bSigma^\dg_\ua\kket{k,k} &= \kket{0,0},~\forall k=1,2,...,d-1.
\end{align}
Before detailing the exact solution, we briefly recall the physical meaning of the relevant parameters. We use $\gamma$ to denote the characteristic relaxation rate, and $\bar n=(e^{\beta\Delta}-1)^{-1}$ to represent the mean thermal boson occupation number at the energy gap $\Delta$. Consequently, the transition rates are given by $\Gamma_\beta(-\Delta)=\gamma(1+\bar n)$ and $\Gamma_\beta(\Delta)=\gamma\bar n$. For conciseness, we adopt the $\gamma$ and $\bar n$ notation in the subsequent calculations.

\subsection{Right eigenvalues and eigenvectors}
The Lindbladian can be written as
\begin{align}
\bL = \sum_{j=1}^{d-1} \gamma(1+\bar n) \bD_\da^{(j)} + \gamma \bar n \bD_\ua^{(j)}.
\end{align}
First, we consider the coherence part.
As the map that we considered is of the Davies form, its coherence part is diagonal in the energy eigenbasis.
In the vectorization representation, we have
\begin{align}
\bL\kket{k, i} = \sum_{j=1}^{d-1} \gamma(1+\bar n) \bD_\da^{(j)} \kket{k,i} = - \gamma(1+\bar n) \kket{k,i},~\forall k\neq i, k>0,i>0.
\end{align}
In addition, we have
\begin{align}
\bL\kket{k, 0} = -\frac{1}{2}\left[\gamma \bar n + (d-1)\gamma(1+\bar n)\right]\kket{k,0}, ~\forall k>0.
\end{align}
Then we consider the population part. 
The population space is spanned by $\{\kket{0,0}, \kket{k,k}, k=1,2,...,d-1\}$.
The image of the map can be written as
\begin{align}
\bD_\ua^{(j)} \kket{0,0} &= \kket{j,j} - \kket{0,0},\\
\bD_\da^{(j)} \kket{0,0} &= 0,\\
\bD_\ua^{(j)} \kket{k,k} &= 0,\\
\bD_\da^{(j)} \kket{k,k} &= \delta_{jk} \left(\kket{0,0} - \kket{k,k}\right).
\end{align}
We then have
\begin{align}
\bL \kket{0,0} &= \sum_{j=1}^{d-1} \gamma \bar n \left(\kket{j,j} - \kket{0,0}\right) = -(d-1)\gamma \bar n \kket{0,0} + \gamma \bar n \sum_{j=1}^{d-1} \kket{j,j},\\
\bL \kket{k,k} &= \sum_{j=1}^{d-1} \gamma(1+\bar n) \delta_{jk} \left(\kket{0,0} - \kket{k,k}\right) = \gamma(1+\bar n) \left(\kket{0,0} - \kket{k,k}\right)\label{eq:g9}.
\end{align}
Let $\kket{S} = \frac{1}{\sqrt{d-1}} \sum_{j=1}^{d-1} \kket{j,j}$ be the uniform combination of all exited states.
We have
\begin{align}
\bL\kket{0,0} &= \sqrt{d-1} \gamma \bar n \kket{S} - (d-1)\gamma \bar n \kket{0,0},\\
\bL\kket{S} &= \gamma(1+\bar n) \sqrt{d-1} \kket{0,0} - \gamma(1+\bar n) \kket{S}.
\end{align}
In the subspace spanned by $\{\kket{0,0}, \kket{S}\}$, the map can be written as a $2\times 2$ matrix:
\begin{align}
M = \begin{pmatrix}
-(d-1)\gamma \bar n & \sqrt{d-1}\gamma(1+\bar n)\\
\sqrt{d-1}\gamma \bar n & -\gamma(1+\bar n)
\end{pmatrix}.
\end{align}
The eigenvalues can be obtained by solving the characteristic polynomial:
\begin{align}
\det(M - \lambda I) = 0.
\end{align}
We have
\begin{align}
\begin{vmatrix}
-(d-1)\gamma \bar n - \lambda & \sqrt{d-1}\gamma(1+\bar n)\\
\sqrt{d-1}\gamma \bar n & -\gamma(1+\bar n) - \lambda
\end{vmatrix} = 0.
\end{align}
By solving the equation, we have
\begin{align}
\lambda = 0,~\text{or}~-\gamma(d \bar n + 1).
\end{align}
The corresponding eigenvectors are
\begin{align}
\kket{r_1} &= \tau_\beta = \exp(-\beta H)/\Tr[\exp(-\beta H)],\\
\kket{r_d} &= \sqrt{d-1} \kket{0,0}-\kket{S}.
\end{align}
Here $\tau_\beta$ is the Gibbs state. 
We here only normalize the eigenvector $\kket{r_d}$ by dividing its norm $\sqrt{d}$, i.e.,
\begin{align}
\kket{r_d} = \sqrt{\frac{d-1}{d}} \kket{0,0} - \sqrt{\frac{1}{d(d-1)}} \sum_{j=1}^{d-1} \kket{j,j}.
\end{align}
By observing \cref{eq:g9}, we can see that the $(d-2)$ eigenvectors in the orthogonal subspace to $\kket{0,0}$ and $\kket{S}$ can be obtained by linearly summing both sides with coefficients $c_j$'s that satisfy $\sum_{j=1}^{d-1} c_j = 0$:
\begin{align}
\bL \sum_{k=1}^{d-1}c_k\kket{k,k} = -\gamma(1+\bar n) \sum_{k=1}^{d-1} c_k \kket{k,k},
\end{align}
where the $\kket{0,0}$ term vanishes as $\kket{S}$ is orthogonal to this subspace.
By normalization, similar to the left eigenvectors, we can choose a convenient basis for this eigenspace:
\begin{align}
\kket{r_k} = \frac{ \kket{1,1} - \kket{k,k} }{ \sqrt{2} }, \quad k = 2, 3, \dots, d-1.
\end{align}
All the right eigenvalues and eigenvectors in the population space are summarized in \Cref{tab:right-eigen-population}.

\begin{table}[h!]
\begin{tabular}{|c|c|c|}
\hline
\text{Eigenvalue} & \text{Multiplicity} & \text{Eigenvector} \\
\hline
$0$ & $1$ & 
$\dfrac{1+\bar{n}}{1+d\bar{n}} \kket{0,0} + \dfrac{\bar{n}}{1+d\bar{n}} \displaystyle\sum_{j=1}^{d-1} \kket{j,j} = \tau_\beta$ \\
\hline
$-\gamma(d \bar{n} + 1)$ & $1$ & 
$\sqrt{\dfrac{d-1}{d}} \kket{0,0} - \dfrac{1}{\sqrt{d(d-1)}} \displaystyle\sum_{j=1}^{d-1} \kket{j,j}$ \\
\hline
$-\gamma(1 + \bar{n})$ & $d-2$ & 
$\dfrac{ \kket{1,1} - \kket{k,k} }{ \sqrt{2} }$ \quad for $k = 2, 3, \dots, d-1$ \\
\hline
\end{tabular}
\caption{Right eigenvectors and eigenvalues of the Lindbladian $\bL$ in population space.}
\label{tab:right-eigen-population}
\end{table}

\subsection{Left eigenvalues and eigenvectors}
We first consider the coherence part.
As the map is diagonal in the coherence part, we have, for $k\neq i, k>0,i>0$,
\begin{align}
\bL^\dg \kket{k,i} = -\gamma(1+\bar n) \kket{k,i}.
\end{align}
We write the adjoint map of $\bD^{(j)}_{\ua,\da}$'s as
\begin{align}
&\bD_\da^{(j)\dg}\kket{k,i} = -\frac{1}{2}\left(\delta_{jk} + \delta_{ji}\right)\kket{k,i},\\
&\bD_\ua^{(j)\dg}\kket{k,i} = 0.
\end{align}
We consider the population space spanned by $\{\kket{0},\kket{k,k}, k\in[1,d-1]\}$.
The map can be written as
\begin{align}
\bL\kket{k,k} &= \gamma (1+\bar n) \kket{0,0} - \gamma (1+\bar n) \kket{k,k},\\
\bL\kket{0,0} &= -\gamma(d-1) \bar n \kket{0,0} + \gamma \bar n \sum_{j=1}^{d-1} \kket{j,j}.
\end{align}
We need to compute the eigenvectors and eigenvalues of the map $\bL^\dg$:
\begin{align}
\bL^\dg \kket{0,0} &= \gamma(1+\bar n) \sum_{j=1}^{d-1} \kket{j,j} - \gamma (d-1) \bar n \kket{0,0},\\
\bL^\dg \kket{k,k} &= -\gamma (1+\bar n) \kket{k,k} + \gamma \bar n \kket{0,0}.
\end{align}
The operator $\bL^\dg$ acts on a $d$-dimensional basis $\{ \kket{0,0}, \kket{1,1}, \dots, \kket{d-1,d-1} \}$. 
Define the symmetric state $\kket{S} = \frac{1}{\sqrt{d-1}} \sum_{j=1}^{d-1} \kket{j,j}$. The action of $\bL^\dg$ is:
\begin{align}
\bL^\dg \kket{0,0} &= -\gamma (d-1) \bar{n} \kket{0,0} + \gamma (1 + \bar{n}) \sqrt{d-1} \kket{S},\\
\bL^\dg \kket{S} &= \gamma \bar{n} \sqrt{d-1} \kket{0,0} - \gamma (1 + \bar{n}) \kket{S}.
\end{align}
The restricted $2\times 2$ matrix in $\{ \kket{0,0}, \kket{S} \}$ is:
$$
M = \begin{pmatrix}
-\gamma (d-1) \bar{n} & \gamma \bar{n} \sqrt{d-1} \\
\gamma (1 + \bar{n}) \sqrt{d-1} & -\gamma (1 + \bar{n})
\end{pmatrix}.
$$
Solve $\det(M - \lambda I) = 0$:
$$
\begin{vmatrix}
-\gamma (d-1) \bar{n} - \lambda & \gamma \bar{n} \sqrt{d-1} \\
\gamma (1 + \bar{n}) \sqrt{d-1} & -\gamma (1 + \bar{n}) - \lambda
\end{vmatrix} = 0.
$$
$$
\left[-\gamma (d-1) \bar{n} - \lambda\right]\left[-\gamma (1 + \bar{n}) - \lambda\right] - \left[\gamma \bar{n} \sqrt{d-1}\right]\left[\gamma (1 + \bar{n}) \sqrt{d-1}\right] = 0.
$$
$$
\lambda^2 + \gamma\left[(d-1)\bar{n} + 1 + \bar{n}\right]\lambda = 0.
$$
Eigenvalues:
$$
\lambda_1 = 0, \quad \lambda_d = -\gamma(d\bar{n} + 1).
$$
For $\lambda_1=0$, the corresponding eigenvector (up to normalization) is
\begin{align}
\kket{l_1} = \kket{0,0} + \sum_{j=1}^{d-1}\kket{j,j} = \kket{\id}.
\end{align}
For $\lambda_d=-\gamma (d\bar n +1)$, we have
\begin{align}
\kket{l_d} = \bar{n}\sqrt{d-1}\kket{0,0} - (1+\bar n)\kket{S}.
\end{align}
The $(d-2)$ eigenvectors in the orthogonal subspace to $\kket{0,0}$ and $\kket{S}$:
$$
\kket{\psi} = \sum_{j=1}^{d-1} c_j \kket{j,j}, \quad \sum_{j=1}^{d-1} c_j = 0, \quad \bL^\dg \kket{\psi} = -\gamma(1 + \bar{n}) \kket{\psi}.
$$
Eigenvalue $\lambda_k = -\gamma(1 + \bar{n})$ with multiplicity $d-2$. Basis:
$$
\kket{\psi_k} = \frac{ \kket{1,1} - \kket{k,k} }{ \sqrt{2} }, \quad k = 2, 3, \dots, d-1.
$$
All the eigenvalues and eigenvectors in the population space are summarized in \Cref{tab:eigen-population}.
\begin{table}[h!]
\begin{tabular}{|c|c|c|}
\hline
\text{Eigenvalue} & \text{Multiplicity} & \text{Eigenvector} \\
\hline
$0$ & $1$ & 
$\dfrac{ \kket{0,0} + \displaystyle\sum_{j=1}^{d-1} \kket{j,j} }{ \sqrt{d} }$ \\
\hline
$-\gamma(d \bar{n} + 1)$ & $1$ & 
$\dfrac{ -\bar{n} \sqrt{d-1} \; \kket{0,0} + \dfrac{1 + \bar{n}}{\sqrt{d-1}} \displaystyle\sum_{j=1}^{d-1} \kket{j,j} }{ \sqrt{ \dfrac{(1 + \bar{n})^2}{d-1} + (d-1)\bar{n}^2 } }$ \\
\hline
$-\gamma(1 + \bar{n})$ & $d-2$ & 
$\dfrac{ \kket{1,1} - \kket{k,k} }{ \sqrt{2} }$ \quad for $k = 2, 3, \dots, d-1$ \\
\hline
\end{tabular}
\caption{Left eigenvectors and eigenvalues of the linear map $\bL$ in population space.}
\label{tab:eigen-population}
\end{table}
Combining with the previous results on the coherence part, we have the following sequence of inequalities:
\begin{align}
0=|\Re(\lambda_1)| \le |\Re(\lambda_2)|=|\Re(\lambda_3)| ... = |\Re(\lambda_{d-1})| \le |\Re(\lambda_c)| \le |\Re(\lambda_d)|.
\end{align}
The slowest decay rate is $\lambda_2 = -\gamma(1+\bar n)$. The corresponding eigenvectors are
\begin{align}
\frac{ \kket{1,1} - \kket{k,k} }{ \sqrt{2} } \quad \text{ for } k = 2, 3, \dots, d-1.
\end{align}
To be orthogonal to the slowest decay mode, the initial state should be either the ground state or the equal superposition of exited states.
The fastest decaying eigenvector can be written as (up to normalization):
\begin{align}
-\bar{n} \sqrt{d-1} \; \kket{0,0} + \dfrac{1 + \bar{n}}{\sqrt{d-1}} \displaystyle\sum_{j=1}^{d-1} \kket{j,j},
\end{align}
which converge to the ground state asymptotically for large $d$.
Thus, we can conclude that the optimal initial state is asymptotically aligned with the fastest decaying mode.

\section{Convergence rate analysis}
\subsection{Bound the ground state overlap by perturbation theory}\label{app:perturbation-theory}
\begin{theorem}\label{thm:bound-perturbed-overlap}
    The optimal probe state for temperature estimation satisfies
    \begin{align}
        \left|\Tr(l^\dg_j \brho^{\star})\right| \le \frac{\varepsilon}{\sqrt{d-1}} \left|\frac{\partial}{\partial\omega}\log \left[e^{\omega \beta} \Gamma_\beta(\omega)\right]\right|_{\omega=\Delta},
    \end{align}
    for all $j$'s that fulfill $|\Re(\lambda_j)|\in (0,(d-1)\Lambda_{\min}/2]$. 
    Here $\Delta:=\omega_1-\omega_0$.
    If the excited states are exactly degenerate ($\epsilon=0$), the conditions are satisfied strictly,
    \begin{align}
        \Tr(l^\dg_j \brho^{\star}) = 0.
    \end{align}
\end{theorem}
\begin{proof}
    Given the unperturbed eigenstate
    \begin{align}
        \bL^\dg\kket{v} = \lambda \kket{v}
    \end{align}
    and the perturbed one
    \begin{align}
        (\bL^\dg + \bL_1^\dg)\left(\kket{v} + \kket{v^1}\right) = (\lambda +\lambda^1) \left(\kket{v} + \kket{v^1}\right),
    \end{align}
    the equation can be further written as
    \begin{align}
        \bL_1^\dg \kket{v} + \bL^\dg \kket{v^1} = \lambda^1 \kket{v} + \lambda \kket{v^1}\label{eq:first-order-perturbation-equation}.
    \end{align}
    We rewrite the Lindblad matrix as
    \begin{align}
        \bL &= \sum_{j=1}^{d-1} J(\omega_j) \left[f_\beta(\omega_j)\bD_\ua^{(j)}+e^{\omega_j \beta} f_\beta(\omega_j)\bD_\da^{(j)}\right]_{\omega_j = \Delta}\\
        &=\sum_{j=1}^{d-1} \Gamma_\beta(\Delta) \bD^{(j)}_{\ua} + \Gamma_\beta '(\Delta) \bD_\da^{(j)}.
    \end{align}
    and the first order perturbation on the energy spectrum can be written as
    \begin{align}
        \bL_1^\dg &= \sum_{j}\partial_{\omega_j}\bL^\dg|_{\omega_j = \Delta} \delta \omega_j = \sum_{j=1}^{d-1}\varepsilon_j \partial_\omega \bL^\dg|_{\omega=\Delta}\\
        &=\sum_{j=1}^{d-1} \varepsilon_j \left\{\partial_\omega\left[J(\omega)f_\beta(\omega)\right] \bD^{(j)}_\ua +\partial_\omega \left[e^{\omega \beta}J(\omega) f_\beta(\omega)\right]\bD_\da^{(j)}\right\}_{\omega=\Delta}.
    \end{align}
    Here $f_\beta(\omega)=[\exp(\omega \beta)\pm 1]^{-1}$ and its derivative is $\partial_\omega f_\beta(\omega) = [\exp(\omega \beta)\pm 1]^{-2}(-1) \exp(\omega\beta)\beta=-\beta\exp(\omega\beta)f_\beta^2(\omega)$.
    We then have
    \begin{align}
        \bL_1^\dg =& \sum_{j=1}^{d-1} \varepsilon_j \left[ 
        \underbrace{
            \left( 
                f_\beta(\Delta) \partial_\omega J(\omega)\big|_{\omega=\Delta} 
                - \beta J(\Delta) e^{\Delta \beta} f_\beta^2(\Delta) 
            \right) 
        }_{\partial_\omega[J(\omega)f_\beta(\omega)]|_{\omega=\Delta}} 
        \bD_\ua^{(j)\dg}
        \right.\nonumber\\
        &+\left. 
        \underbrace{
            \left( 
                \beta e^{\Delta \beta} J(\Delta) f_\beta(\Delta) 
                + e^{\Delta \beta} f_\beta(\Delta) \partial_\omega J(\omega)\big|_{\omega=\Delta} 
                - \beta e^{2\Delta \beta} J(\Delta) f_\beta^2(\Delta) 
            \right) 
        }_{\partial_\omega[e^{\omega\beta}J(\omega)f_\beta(\omega)]|_{\omega=\Delta}} 
        \bD_\da^{(j)\dg}
    \right]\\
    =& \sum_{j=1}^{d-1} \varepsilon_j \left[\dot\Gamma_\beta(\Delta) \bD_\ua^{(j)\dg} + \dot\Gamma_\beta '(\Delta) \bD_\da^{(j)\dg}\right].
    \end{align}
    First, let $\kket{v}$ be the slowest decaying mode in the population space.
    As the unperturbed Lindblad operator is degenerated, we need to leave the unperturbed eigen-operator in a generic form, i.e., $\kket{v} = \sum_{j=1}^{d-1} v_j \kket{j,j}$. According to \Cref{app:satisfy-SAT}, the slowest decaying modes in the population space must satisfy $\bbkk{S}{v}=0$.
    In addition, we remark that the left eigen-operator should be normalized according to the norm of its corresponding right eigen-operator, i.e., $\bbkk{\hat l_j}{\hat r_j}=1$.
    However, the left and right eigen-operators coincides within this degenerated subspace, thereby they can be normalized by themselves.
    More specifically, we have the following constraints on the coefficients $\{v_j\}$:
    \begin{align}
        \sum_{j=1}^{d-1} v_j = 0,~ \sum_{j=1}^{d-1} |v_j|^2 = 1.\label{eq:unperturbed-coeff-constraints}
    \end{align}
    We also let the first-order perturbation be of the general form as well, i.e., $\kket{v^1} = \sum_{j=0}^{d-1} v_j^1 \kket{j,j}$.
    Our target is to bound $|v_0^1|$.
    
    We now evaluate each term in \Cref{eq:first-order-perturbation-equation}.
    First, we have
    \begin{align}
        \bL_1^\dg \kket{v} &= \sum_{j=1}^{d-1} v_j \bL_1^\dg \kket{j,j} = \sum_{j,k=1}^{d-1} v_k \varepsilon_j \left[\dot{\Gamma}_\beta (\Delta) \bD_\ua^{(j)\dg} \kket{k,k}+ \dot{\Gamma}_\beta' (\Delta)\bD_\da^{(j)\dg} \kket{k,k}\right]\\
        &= \sum_{j,k=1}^{d-1} v_k \varepsilon_j \left[\dot{\Gamma}_\beta (\Delta) \delta_{j,k}\kket{0,0}+ \dot{\Gamma}_\beta' (\Delta)(-\delta_{j,k})\kket{k,k}\right]\\
        &= \sum_{j=1}^{d-1} v_j \varepsilon_j \left[\dot{\Gamma}_\beta (\Delta) \kket{0,0}-\dot{\Gamma}_\beta' (\Delta)\kket{j,j}\right].
    \end{align}
    Its projection onto $\bbra{k,k}$ can be written as
    \begin{align}
        \bbra{k,k} \bL_1^\dg \kket{v} = -v_k \varepsilon_k \dot{\Gamma}_\beta'(\Delta)\label{eq:first-order-perturb-first-term}.
    \end{align}
    Then we evaluate the second term in \Cref{eq:first-order-perturbation-equation}:
    \begin{align}
        \bL^\dg \kket{v^1} &= v_0^1 \bL^\dg\kket{0,0} + \sum_{j=1}^{d-1} v_j^1 \bL^\dg \kket{j,j} \\
        &=v_0^1 \bL^\dg\kket{0,0} + \sum_{j,k=1}^{d-1} v_k^1 \left[\Gamma_\beta(\Delta)\bD^{(j)\dg}_{\ua}\kket{k,k}+\Gamma_\beta '(\Delta)\bD_\da^{(j)\dg}\kket{k,k}\right]\\
        &= v_0^1 \left[\Gamma_\beta(\Delta) \bSigma_\ua^\dg\kket{0,0} + \Gamma'_\beta(\Delta)\Sigma_\da^\dg\kket{0,0}\right] +  \sum_{k=1}^{d-1}v_k^1 \left[\Gamma_\beta(\Delta) \bSigma_\ua^\dg\kket{k,k} + \Gamma'_\beta(\Delta)\Sigma_\da^\dg\kket{k,k}\right]\\
        &= v_0^1 \left[-\Gamma_\beta(\Delta) (d-1)\kket{0,0} -\Gamma'_\beta (\Delta) \kket{0,0} + \Gamma'_\beta (\Delta) \sum_{j=1}^{d-1} \kket{j,j}\right]\nonumber \\
        &\qquad +  \sum_{k=1}^{d-1}v_k^1 \left[\Gamma_\beta(\Delta) \kket{0,0} - \Gamma'_\beta(\Delta)\kket{k,k}\right].
    \end{align}
    We calculate its projection onto $\bbra{k,k}$, we have
    \begin{align}
        \bbra{k,k}  \bL^\dg \kket{v^1} &= v_0^1 \Gamma'_\beta(\Delta) - v_k^1 \Gamma'_\beta(\Delta)\label{eq:first-order-perturb-second-term}.
    \end{align}
    The right-hand side of \Cref{eq:first-order-perturbation-equation} is simple and we can directly write down its projection onto $\bbra{k,k}$ as:
    \begin{align}
        \lambda^1 v_k + \lambda v_k^1\label{eq:first-order-perturb-right-hand-side}.
    \end{align}
    According to \Cref{app:satisfy-SAT}, we have $\lambda = -\Gamma'_\beta(\Delta)$.
    Combining \Cref{eq:first-order-perturbation-equation,eq:first-order-perturb-first-term,eq:first-order-perturb-second-term,eq:first-order-perturb-right-hand-side}, we have
    \begin{align}
        -v_k \varepsilon_k \dot{\Gamma}_\beta'(\Delta) + v_0^1 \Gamma'_\beta(\Delta) - v_k^1 \Gamma'_\beta(\Delta) = \lambda^1 v_k - \Gamma'_\beta(\Delta) v_k^1,
    \end{align}
    or equivalently,
    \begin{align}
        -v_k \varepsilon_k \dot{\Gamma}_\beta'(\Delta) + v_0^1 \Gamma'_\beta(\Delta)= \lambda^1 v_k.
    \end{align}
    Taking the summation over $k$ on both sides, the term $\sum_{k=1}^{d-1} v_k$ vanishes, i.e.,
    \begin{align}
        -\dot{\Gamma}_\beta'(\Delta)\sum_{k=1}^{d-1}v_k \varepsilon_k  + (d-1)v_0^1 \Gamma'_\beta(\Delta)= \lambda^1\underbrace{\sum_{k=1}^{d-1} v_k}_{=0}.
    \end{align}
    We then have
    \begin{align}
        v_0^1 &= \frac{\dot{\Gamma}_\beta'(\Delta)}{(d-1)\Gamma'_\beta(\Delta)} \sum_{k=1}^{d-1} v_k \varepsilon_k.
    \end{align}
    By Cauchy–Schwarz inequality, we have
    \begin{align}
        |v_0^1| &\le \left|\frac{\dot{\Gamma}_\beta'(\Delta)}{(d-1)\Gamma'_\beta(\Delta)}\right| \sqrt{\left(\sum_{j=1}^{d-1} |v_j|^2\right)\left(\sum_{j=1}^{d-1} |\varepsilon_j|^2\right)}\\
        &\le \left|\frac{\dot{\Gamma}_\beta'(\Delta)}{(d-1)\Gamma'_\beta(\Delta)}\right| \sqrt{(d-1)\varepsilon^2}\\
        &= \left|\frac{\dot{\Gamma}_\beta'(\Delta)}{\Gamma'_\beta(\Delta)}\right| \frac{\varepsilon}{\sqrt{d-1}}\\
        &=  \frac{\varepsilon}{\sqrt{d-1}}\left|\frac{\partial}{\partial \omega} \log \Gamma'_\beta(\omega)\right|_{\omega=\Delta},
    \end{align}
    which completes the proof.
\end{proof}

\subsection{Bound the convergence rate}
The system Hamiltonian {comprises} $d-1$ non-degenerate states with {an} energy difference bounded by $\varepsilon${; consequently,} the eigenoperators of the Liouvillian are not orthogonal.
However, {two facts} allow us to obtain a tight bound on the convergence rate.
First, according to \cref{thm:bound-perturbed-overlap}, the optimal probe state has {a} small overlap with the slowest decaying modes.
Second, the right- (or left-) eigenoperators $\hat r_j$ satisfy {approximate} orthogonality, {where} $\Tr(\hat r_i ^\dg \hat r_j)=\delta_{ij}+O(\varepsilon)$. 
{This leads to} the following lemma.
\begin{lemma}\label{lem:thermal-convergence-bound}
    {There exists a constant} $t'>0$ {such that the following bound holds for} $t\ge t'$:
    \begin{align}
        \fnorm{\rho_t^\star - \tau_\beta}^2 \le \frac{11}{10} \varepsilon^2 \exp(-2 \Lambda_{\min} t)\left(\frac{d-2}{d-1}\right)\left|\frac{\partial}{\partial\omega}\log\left[e^{\omega\beta}\Gamma_\beta(\omega)\right]\right|_{\omega = \Delta}^2,
    \end{align}
    {where} $\rho_t^\star:=\exp(\cL t)[\rho^\star]${.}
\end{lemma}
\begin{proof}
    {Expanding the square of the Frobenius distance yields}
    \begin{align}
        \fnorm{\rho_t -\tau_\beta}^2 &= \sum_{i,j:\lambda_i,\lambda_j\in\mathrm{Spec(\cL)}\backslash\{0\}} e^{(\lambda_i^* + \lambda_j)t}\Tr(\rho_0\hat l_i)\Tr(\hat l_j^\dg \rho_0) \Tr(\hat r_i^\dg \hat r_j)\\
        &= \sum_{i,j:\lambda_i,\lambda_j \in \mathrm{Spec}_\mathcal{P}(\cL)\backslash\{0\}}e^{(\lambda_i^* + \lambda_j)t}\Tr(\rho_0\hat l_i)\Tr(\hat l_j^\dg \rho_0) \Tr(\hat r_i^\dg \hat r_j) + \sum_{k:\lambda_k\in\mathrm{Spec}_\mathcal{C}(\cL)} e^{2\Re(\lambda_k)t} \left|\Tr(\hat l_k^\dg\rho_0)\right|^2.
    \end{align}
    {While the cross terms are structurally intricate, this expression simplifies significantly in the asymptotic limit. For sufficiently large time} $t${, the dynamics are strictly dominated by the terms corresponding to the eigenvalues with the smallest non-zero real parts, which govern the slowest decay rates.}
    
    {Specifically, by substituting the initial state} $\rho_0 = \rho^\star${, any components originating from the coherent subspace} $\mathcal{C}$ {identically vanish, as} $\rho^\star$ {possesses strictly zero overlap with} $\mathcal{C}${.}
    
    {Consequently, the remaining non-zero contributions must reside entirely within the population subspace} $\mathcal{P}${. For a sufficiently large time} $t\ge t'${, we can introduce a prefactor} $M>1$ {to safely upper-bound the summation by absorbing the exponentially suppressed faster-decaying modes:}
    \begin{align}
        \fnorm{\rho_t -\tau_\beta}^2 \le M\sum_{i,j:\lambda_i,\lambda_j \in \mathrm{Lower~tail~of~Spec}_\mathcal{P}(\cL)\backslash\{0\}}e^{(\lambda_i^* + \lambda_j)t}\Tr(\rho_0\hat l_i)\Tr(\hat l_j^\dg \rho_0) \Tr(\hat r_i^\dg \hat r_j).
    \end{align}
    {Although} $M$ {can be chosen to be arbitrarily close to} $1$ {as} $t \to \infty${, we rigorously fix} $M = 11/10$ {to ensure the bound strictly holds for all} $t \ge t'${.}
    
    {Furthermore, recall that the overlaps between} $\rho^\star$ {and the eigenoperators at the lower tail of the spectrum are strictly upper-bounded by \cref{thm:bound-perturbed-overlap}. This allows us to establish the following inequality:}
    \begin{align}
        \fnorm{\rho_t - \tau_\beta}^2 \le \frac{11}{10}\sum_{j:\lambda_j \in \mathrm{Lower~tail~of~Spec}_\mathcal{P}(\cL)\backslash\{0\}} e^{2\Re(\lambda_j) t} \left(\frac{\varepsilon^2}{d-1}\right) \left|\frac{\partial}{\partial\omega}\log\left[e^{\omega\beta}\Gamma_\beta(\omega)\right]\right|_{\omega = \Delta}^2,
    \end{align}
    {where the higher-order correction} $O(\varepsilon^3)${, arising from the approximate orthogonality condition} $\Tr(\hat r^\dg_i \hat r_j)=\delta_{ij}+O(\varepsilon)${, is comfortably absorbed by our conservatively chosen prefactor} $M${.}
    
    {From the spectral properties of} $\cL${, the multiplicity of these slow-decaying modes is exactly given by}
    \begin{align}
        \left| \mathrm{Lower~tail~of~Spec}_\mathcal{P}(\cL)\backslash\{0\}\right| = d-2.
    \end{align}
    
    {Substituting this multiplicity, the bound becomes:}
    \begin{align}
        \fnorm{\rho_t - \tau_\beta}^2 \le M \varepsilon^2 \exp(-2 \Lambda_{\min} t)\left(\frac{d-2}{d-1}\right)\left|\frac{\partial}{\partial\omega}\log\left[e^{\omega\beta}\Gamma_\beta(\omega)\right]\right|_{\omega = \Delta}^2,
    \end{align}
    {which, upon substituting} $M=11/10${, immediately completes the proof.}
\end{proof}

\subsection{Probability of exceeding a random state}

\begin{theorem}[Exponential high-probability exceeding]\label{thm:prob-of-exceeding-app}
{Define} $\rho_t = e^{\cL t}[\rho_0]$ {as} the evolved state starting from a random initial state
\[
\rho_0=(1-\alpha)\tau_\beta+\alpha\sigma,
\qquad \alpha\in(0,1],
\qquad \sigma = U\ketbra{0}{0}U^\dagger,\ U\sim\mu_H.
\]
Let $\rho_t^\star := e^{\cL t}[\rho^\star]$.
{We adopt} the model and notation {from} the main text ({specifically, the} Davies form, excited-state detunings $|\varepsilon_j|\le \varepsilon$, and $d\ge 3$).
Then there exists a (model-dependent) time $t'>0$ such that, with probability at least $1-\delta_{\exp}$ over $U\sim\mu_H$, $\rho_t^\star$ {converges faster than} $\rho_t$ {, satisfying}
\[
\forall t\ge t',\qquad \fnorm{\rho_t^\star-\tau_\beta}\le \fnorm{\rho_t-\tau_\beta}.
\]
Moreover, writing
\[
\mu_d := \frac{(d-1)(d-2)}{d(d+1)},
\qquad 
g:=\left|\partial_\omega \log\!\left[e^{\omega\beta}\Gamma_\beta(\omega)\right]\right|_{\omega=\omega_1-\omega_0},
\]
the failure probability obeys the exponential bound
\begin{align}
\delta_{\exp}
\;\le\;
2\exp\!\left(
-\frac{d}{36\pi^3}\,
\Big[\mu_d-{\frac{11}{10}}\tfrac{\varepsilon^2}{\alpha^2}g^2\Big]_+^{\,2}
\right),
\label{eq:delta-exp-bound}
\end{align}
where {$11/10$} is the constant {established} in Lemma~\ref{lem:thermal-convergence-bound} and $[x]_+:=\max\{x,0\}$.
In particular, whenever $\mu_d-{\frac{11}{10}}\frac{\varepsilon^2}{\alpha^2}g^2=\Omega(1)$, one has $\delta_{\exp}=\exp(-\Omega(d))$.
\end{theorem}

\begin{proof}
Fix $U$ and {define} $\sigma=\ket{\psi}\!\bra{\psi}$ with $\ket{\psi}=U\ket{0}$.
{Following the previous methodology}, we compare the {asymptotic} coefficients of the slowest decay modes.

\smallskip\noindent
By Lemma~\ref{lem:thermal-convergence-bound}, there {exists} $t'>0$ {such that} for all $t\ge t'$, {using the previously derived factor of $11/10$,}
\begin{align}
\fnorm{\rho_t^\star-\tau_\beta}^2
\;\le\;
{\frac{11}{10}}\,\varepsilon^2\,e^{-2\Lambda_{\cL}^{\min}t}\,
\frac{d-2}{d-1}\,g^2.
\label{eq:opt-upper}
\end{align}
{Conversely}, since $\rho_t=(1-\alpha)\tau_\beta+\alpha e^{\cL t}[\sigma]$,
\[
\fnorm{\rho_t-\tau_\beta}^2=\alpha^2\fnorm{e^{\cL t}[\sigma]-\tau_\beta}^2.
\]
{Since} $\cL$ is {in} Davies {form} and the coherence subspace is normal (hence {admitting} an orthogonal eigen-operator decomposition), the squared Frobenius distance {incorporates} a {non-negative} contribution from the coherence modes that (in the unperturbed {or} near-degenerate setting) sit in the lower tail of $\mathrm{Spec}_{\cC}(\cL)$.
{Consequently}, for all $t\ge 0$,
\begin{align}
\fnorm{\rho_t-\tau_\beta}^2
\;\ge\;
\alpha^2 e^{-2\Lambda_{\cL}^{\min}t}\, f(\psi),
\label{eq:rand-lower}
\end{align}
where
\begin{align}
f(\psi)
\;:=\;
\sum_{\substack{i,j=1\\i\neq j}}^{d-1}
\left|\bra{i}\sigma\ket{j}\right|^2
=
\sum_{\substack{i,j=1\\i\neq j}}^{d-1}
|\psi_i|^2|\psi_j|^2
\in[0,1].
\label{eq:def-f}
\end{align}
Combining {Inequalities}~\eqref{eq:opt-upper} and~\eqref{eq:rand-lower}, we {observe} that the event
\begin{align}
f(\psi)\;>\;\theta
\qquad\text{where}\qquad
\theta:={\frac{11}{10}}\frac{\varepsilon^2}{\alpha^2}\frac{d-2}{d-1}g^2,
\label{eq:good-event}
\end{align}
implies that for all $t\ge t'$,
\[
\fnorm{\rho_t-\tau_\beta}^2
\;\ge\;
\alpha^2 e^{-2\Lambda_{\cL}^{\min}t} f(\psi)
\;>\;
{\frac{11}{10}}\varepsilon^2 e^{-2\Lambda_{\cL}^{\min}t}\frac{d-2}{d-1}g^2
\;\ge\;
\fnorm{\rho_t^\star-\tau_\beta}^2,
\]
{demonstrating that} $\rho_t^\star$ {converges more rapidly than} $\rho_t$.
{Thus, we bound the failure probability as}
\begin{align}
\delta_{\exp}
\;\le\;
\Pr_{U\sim\mu_H}\!\big[f(\psi)\le \theta\big].
\label{eq:delta-reduction}
\end{align}

\smallskip\noindent
For {a} Haar-random {state} $\ket{\psi}\in\mathbb{C}^d$, {we utilize} the standard moment identity
\[
\mathbb{E}\big[|\psi_i|^2|\psi_j|^2\big]=\frac{1}{d(d+1)}
\qquad (i\neq j).
\]
Hence, by~\eqref{eq:def-f},
\begin{align}
\mathbb{E}_{U\sim\mu_H} f(\psi)
=
\sum_{\substack{i,j=1\\i\neq j}}^{d-1}\frac{1}{d(d+1)}
=
\frac{(d-1)(d-2)}{d(d+1)}
=
\mu_d.
\label{eq:Ef}
\end{align}

Next, {we consider} $f$ as a function on the unit sphere $S^{2d-1}$ ({identifying} $\mathbb{C}^d\simeq\mathbb{R}^{2d}$).
Let $\ket{\psi}$ {and} $\ket{\phi}$ be unit vectors{, and define} $\sigma_\psi=\ket{\psi}\!\bra{\psi}$ {and} $\sigma_\phi=\ket{\phi}\!\bra{\phi}$.
Let $\mathcal{P}_\cC$ denote the orthogonal projector (w.r.t.\ the Hilbert--Schmidt inner product) onto the operator subspace
$\mathrm{span}\{\ket{i}\!\bra{j}: i,j\in\{1,\dots,d-1\},\, i\neq j\}$.
{It follows that} $f(\psi)=\|\mathcal{P}_\cC(\sigma_\psi)\|_2^2$. Therefore,
\begin{align*}
|f(\psi)-f(\phi)|
&=
\Big|
\|\mathcal{P}_\cC(\sigma_\psi)\|_2^2-\|\mathcal{P}_\cC(\sigma_\phi)\|_2^2
\Big|\\
&\le
\big(\|\mathcal{P}_\cC(\sigma_\psi)\|_2+\|\mathcal{P}_\cC(\sigma_\phi)\|_2\big)\,
\|\mathcal{P}_\cC(\sigma_\psi-\sigma_\phi)\|_2\\
&\le 2\,\|\sigma_\psi-\sigma_\phi\|_2.
\end{align*}
Using $\|\sigma_\psi-\sigma_\phi\|_2^2=2-2|\langle\psi|\phi\rangle|^2\le 2\|\psi-\phi\|_2^2$, we {obtain}
\[
|f(\psi)-f(\phi)|\le 2\sqrt{2}\,\|\psi-\phi\|_2.
\]
{Consequently,} $f$ is $L$-Lipschitz on $S^{2d-1}$ with $L=2\sqrt{2}$.

\smallskip\noindent
A standard {formulation} of Lévy's lemma (e.g.\ {, see} \cite{mele2024IntroductionHaarMeasureToolsQuantumInformation}) states that for an $L$-Lipschitz function $h:S^{n}\to\mathbb{R}$,
\[
\Pr\big(|h-\mathbb{E}h|\ge a\big)\le 2\exp\!\left(-\frac{(n+1)a^2}{9\pi^3 L^2}\right).
\]
{Applying} this {inequality} with $h=f$, $n=2d-1$, $L=2\sqrt{2}$, and $a:=\mu_d-\theta${, we proceed based on the sign of} $a$.
{For} $a>0$, we have
\[
\Pr[f(\psi)\le \theta]
\le
\Pr\big(|f(\psi)-\mu_d|\ge \mu_d-\theta\big)
\le
2\exp\!\left(-\frac{2d\,(\mu_d-\theta)^2}{9\pi^3\cdot 8}\right)
=
2\exp\!\left(-\frac{d}{36\pi^3}(\mu_d-\theta)^2\right).
\]
{For} $a\le 0${,} the bound is trivial (hence the {use of the} $[\cdot]_+$ notation).
Combining {this result} with~\eqref{eq:delta-reduction} and substituting~\eqref{eq:good-event} completes the proof and yields {the exponential bound in}~\eqref{eq:delta-exp-bound}.
\end{proof}
\end{widetext}
\end{document}